\def\useieeelayout{1}
\def\review{0} 

\if\useieeelayout1
\documentclass[journal,twoside,web]{ieeecolor}
\usepackage{generic}
\else
\documentclass{article}
\usepackage{parskip}
\usepackage{fullpage}
\usepackage{amsthm}
\usepackage{color}
\fi

\usepackage{cite}
\usepackage{amsmath,amssymb,amsfonts}
\usepackage{algorithmic}
\usepackage{graphicx}
\usepackage{algorithm,algorithmic}
\usepackage{hyperref}
\hypersetup{hidelinks}
\usepackage{textcomp}
\usepackage{caption}
\usepackage{subcaption}
\usepackage{stmaryrd}
\UseRawInputEncoding

\if\useieeelayout0 \if\review1 
\usepackage{setspace}
\usepackage[mathlines]{lineno}
\fi \fi

\newtheorem{theorem}{Theorem}

\newtheorem{lemma}{Lemma}
\newtheorem{proposition}{Proposition}
\newtheorem{definition}{Definition}
\newtheorem{assumption}{Assumption}

\newtheorem{remark}{Remark}

\def\mytitle{Data-driven memory-dependent abstractions of dynamical systems via a Cantor-Kantorovich metric}

\if\useieeelayout1

\else

\fi

\def\BibTeX{{\rm B\kern-.05em{\sc i\kern-.025em b}\kern-.08em
    T\kern-.1667em\lower.7ex\hbox{E}\kern-.125emX}}
\begin{document}
\title{\mytitle}

\author{Adrien Banse, Licio Romao, Alessandro Abate and Rapha\"el M. Jungers
\thanks{This project has received funding from the European Research Council (ERC) under the European Union’s Horizon 2020 research and innovation program under grant agreement No 864017 (L2C). A. Banse is supported by the French Community of Belgium in the framework of a FNRS/FRIA grant and R. M. Jungers is a FNRS honorary Research Associate.}
\thanks{A. Banse and R. M. Jungers are with the ICTEAM, UCLouvain, Louvain-la-Neuve, Belgium (e-mails: \{adrien.banse, raphael.jungers\}@uclouvain.be). L. Romao is with the Department of Aeronautics \& Astronautics, Stanford University, Stanford, USA (e-mail: licio@stanford.edu). A. Abate is with the Department of Computer Science, University of Oxford, Oxford, UK (e-mail: alessandro.abate@cs.ox.ac.uk).}
}

\maketitle

\begin{abstract}
    Abstractions of dynamical systems enable their verification and the design of feedback controllers using simpler, usually discrete, models. In this paper, we propose a data-driven abstraction mechanism based on a novel metric between Markov models. Our approach is based purely on observing output labels of the underlying dynamics, thus opening the road for a fully data-driven approach to construct abstractions. Another feature of the proposed approach is the use of memory to better represent the dynamics in a given region of the state space. We show through numerical examples the usefulness of the proposed methodology.
\end{abstract}

\if\useieeelayout1
\else
\fi

\section{Introduction}
\label{sec:introduction}

\if\useieeelayout0 \if\review1 
\linenumbers 
\doublespacing
\fi \fi

The complexity of dynamical systems emerging from several industrial applications has dramatically increased in the past years, which raises additional challenges for their analysis and control \cite{Lee2015, Kim2012, Alur2015,Lee2016}. \emph{Abstraction techniques} provide a way to tame complexity in the verification and/or control design step by producing a (usually discrete) representation of the underlying dynamical system (see e.g. \cite{Tabuada2009,Badings2023,Badings2023b,pmlr-v211-banse23a,Lavaei2022, vanderSchaft2004}). The resulting abstract models are then used to indirectly verify the concrete model, or to design feedback controllers by means of a procedure referred to as control refinement \cite{HSA17}.


Two of the main limitations for creating abstractions of dynamical systems are the curse of dimensionality and the reliance on the mathematical model representing the dynamics. The first limitation results from the partitioning of the state space of the concrete model into a finite set of blocks. For a given accuracy, the number of finite states grows exponentially with the dimension of the initial system (see e.g. \cite{Reissig2014}). This leads to discrete representations that, in order to meet specific accuracy levels, require prohibitive computational resources. For this reason, model-based methods to partition the state space in a smarter fashion have appeared in the literature \cite{SA13}. Amongst them, \emph{memory-dependent} abstractions have been used in \cite{Schmuck2014, Majumdar2020} to mitigate this issue. The second limitation has attracted much attention in recent years in the control community: in many applications, the system to control is too complicated for a model-based analysis, or a given mechanistic model is not even available. Therefore, there is a growing need for \emph{data-driven} control approaches, namely techniques that are purely based on the observation of data from the system, thus without resorting to the cumbersome and potentially error-prone work of building a model for the system \cite{DePersis2020,Wang2021,Berberich2021}.

Our new approach, presented in this paper, aims at mitigating both of these limitations, by enabling the construction of a Markov model purely based on output data and allowing for a smart and frugal
refinement of the state-space partitions. To enrich the representation of our discrete model, we consider memory-dependent Markov models, and propose an algorithmic procedure to further enrich the model in specific regions of the state space. Previous work such as \cite{Coppola2023, Coppola2024} also combine memory-dependent and data-driven approaches, but does not profit from an adaptive approach and use deterministic abstractions. Other previous works, such as \cite{Badings2023,Badings2023b}, use stochastic models such as Interval Markov Decision Processes to circumvent the uncertainty on the concrete system, but do not leverage memory.

To enable a non-uniform partitioning of the state space of the dynamics, we develop a novel notion of \emph{metric between Markov models} that relies on the theory of optimal transport. More precisely, we leverage the Kantorovich metric in order to evaluate the difference between two Markov models in terms of the probability distributions that they define on the output language. To define this optimal transport metric, we equip the space of words with the Cantor distance\footnote{In order to ease the reading of this paper, we use the term \emph{metric} for the Kantorovich metric and the term \emph{distance} for the underlying Cantor distance.}. Kantorovich metrics for Markov models have already been studied in \cite{Desharnais2004, Breugel2007, Madras2010, Rudolf2018}, but with different underlying metrics. The Cantor distance has been widely investigated in the field of symbolic dynamics \cite{Lind1995, Fogg2002, Rozenberg2002}, leading to interesting results about its topological structure. We also present a new algorithm that efficiently computes the proposed metric between Markov models.


\textbf{Contributions}\footnote{Preliminary results have been presented in \cite{Banse2023}, but without proofs. Besides, this work  considers a more general model for abstracting dynamical systems, and presents a new, thorough numerical example.} We develop a new framework to abstract dynamical systems. The proposed abstraction technique leverages memory to implicitly build a non-uniform partition of the state space, hence enabling the refinement of the abstraction in regions where the underlying dynamics has more intricate behaviours. A second contribution consists in proposing a new metric between Markov models based on optimal transport, and showing a recursive algorithm for its computation that improves upon a na\"ive linear programming formulation. With that, we also contribute to well established research efforts outside of the control community (see e.g. \cite{Desharnais2004, Garcia2022, Chen2014,Rached2004}).

\textbf{Outline}\,\, The rest of this paper is organised as follows. In Section~\ref{sec:preli}, we provide theoretical background for dynamical systems and labeled Markov chains. In Section~\ref{sec:abstractions}, we introduce adaptive memory abstractions, and show that they satisfy behavioural inclusion (see \cite{Tabuada2009}).
We provide a method to choose a convenient partitioning given a metric between Markov chains, and prove that it always yields well-defined abstractions. In Section~\ref{sec:ck-distance}, we introduce a specific Cantor-Kantorovich metric between Markov chains. We prove that it is indeed a metric, and we give an algorithm to approximate it efficiently. In Section~\ref{sec:application}, we illustrate our procedure with the Cantor-Kantorovich metric on a numerical example. We finally conclude in Section~\ref{sec:conclusions}.

\textbf{Notations} In this work, $\mathbb{R}$ is the set of reals, $\mathbb{N}$ is the set of natural and $\mathbb{N}_{\geq 0}$ is the set of non-negative natural numbers. Given a finite alphabet $A$, a \emph{word} $w \in A^n$ is denoted $w = a_1 \dots a_n$. Given $s, t \in \mathbb{N}_{\geq 0}$, a \emph{timed word}, denoted $w_{[-s, t]}$, is a couple $w \in A^{s+t+1}$ and time interval $[-s, t]$. It is denoted $w_{[-s, t]} = a_{-s}\dots a_{t}$ with small abuse of notation. Given a set $X$, its \emph{Kleene closure} is noted $X^*$, and the $i$-th \emph{functional power} of a function $f$ is noted $f^i$. In terms of computational complexity, we say that $f(n) = \mathcal{O}(g(n))$ if there exists $n_0$ and $c > 0$ such that, for all $n \geq n_0$, $|f(n)| \leq cg(n)$, and we say that $f(n) = \Omega(g(n))$ if there exists $n_0$ and $c > 0$ such that, for all $n \geq n_0$, $f(n) \geq cg(n)$.

\section{Preliminaries} \label{sec:preli}

In this section, we formally define dynamical systems and labeled Markov chains. We also introduce a notion of \emph{behaviour} for both models.
 
\subsection{Dynamical systems}

A \emph{dynamical system} is the 4-tuple $\Sigma = (X, A, f, h)$ that defines the relation
\begin{equation} \label{eq:det_sys}
    x_{k+1} = f(x_k), \quad y_k = h(x_k), 
\end{equation}
where $X \subseteq \mathbb{R}^d$ is the \emph{state space}, $A$ is a finite alphabet called the \emph{output space}, $f : X \to X$ is an invertible \emph{transition
function}, and $h : X \to A$ is the \emph{output function}. The variables $x_k$ are called \emph{states}, and variables $y_k$ are called \emph{outputs} at time $k$. 

For $s, t \in \mathbb{N}$, given a word $w_{[-s, t]} = a_{-s}\dots a_{t}$ with $a_i \in A$ for all $i = -s, \dots, t$, one can define a subset of the state space as follows
\begin{equation}
    \left[w_{[-{s}, t]}\right] = \{x \in X : h(f^{i}(x)) = a_{i} \, \forall i = -s, \dots, t\}.
\end{equation}
For example, the subset $\left[011_{[-1, 1]}\right]$ corresponds to the set of states $x \in X$ such that $h(f^{-1}(x)) = 0$, $h(x) = 1$ and $h(f(x)) = 1$. A set of words 
\begin{equation}
    W = \left\{w^1_{[-s_1, t_1]}, \dots, w^k_{[-s_k, t_k]}\right\}
\end{equation}
defines a \emph{partition} of the state space if 
\begin{gather}
    \bigcup_{i = 1, \dots, k} \left[w^{i}_{[-s_i, t_i]}\right] = X, \label{eq:partition_cond_1}  \\
    \forall i \neq j, \, \left[w^{i}_{[-s_i, t_i]}\right] \cap \left[w^{j}_{[-s_j, t_j]}\right] = \emptyset. \label{eq:partition_cond_2}
\end{gather}
In this case, a subset $[w_{[-s, t]}]$ is called a \emph{block}.

In this paper, we consider data-driven applications in which the initial state of a dynamical system is sampled. We assume that the state space $X$ is endowed with a measure $\lambda$ on the probability space $(X, \mathcal{B}(X), \lambda)$, where $\mathcal{B}(X)$ is the Borel $\sigma$-algebra generated by the topology of $X$. In other words, we consider that the system \eqref{eq:det_sys} is such that
\begin{equation}
    x_0 \sim \lambda.
\end{equation}

For clarity, for a given word $w_{[-s, t]}$, we call the quantity $\lambda\left(\left[ w_{[-s, t]} \right]\right)$ the \emph{probability} of observing the sequence $y_{-s}\dots y_t = w$, which we denote as $\mathbb{P}_{\lambda}(y_{-s}\dots y_{t} = w)$. Similarly, given two words $w_{[-s, t]}$ and $w'_{[-s', t']}$, we call the quantity 
\begin{equation}
    \lambda\left(f^{-1}\left(\left[ w'_{[-s', t']} \right]\right) \cap \left[ w_{[-s, t]} \right]\right) / \lambda\left(\left[ w_{[-s, t]} \right]\right)
\end{equation}
the \emph{conditional probability} of observing $y_{-s'+1}\dots y_{t'+1} = w'$ knowing that $y_{-s}\dots y_{t} = w$, which we denote by  $\mathbb{P}_{\lambda}(y_{-s'+1}\dots y_{t'+1} = w' | y_{-s}\dots y_{t} = w)$. 

We now introduce the notion of finite \emph{behaviour} for this class of dynamical systems. The behaviour of a dynamical system $\Sigma$ with initial measure $\lambda$, denoted as $B(\Sigma, \lambda) \subseteq A^*$, is the set of finite sequences $w \in A^k$, for all $k \in \mathbb{N}$, such that $\mathbb{P}_{\lambda}(y_0 \dots y_{k-1} =w) > 0$. In other words, the behaviour of $\Sigma$ contains the set of words that can be reached with positive probability.

\subsection{Labeled Markov chains}

A labeled Markov chain is defined as a 5-tuple $\Gamma = (S, A, \tau, \mu, l)$ where $S$ is a finite set of \emph{states}, $A$ is a finite \emph{alphabet}, $\tau$ is the \emph{transition matrix} defined on $S \times S$, $\mu$ is an \emph{initial measure} defined on $S$ and $l : S \to A$ is a \emph{labeling function}. For two states $s, s' \in S$, the entry $\tau_{s, s'}$ of the transition matrix is defined as 
\begin{equation}
    \tau_{s, s'} = \mathbb{P}(X_{k+1} = s' \, | \, X_k = s), 
\end{equation}
where $X_1, X_2, \dots$ is a sequence of random variables.

Consider the equivalence relation on $S$ defined as $s \sim s'$ if and only if $l(s) = l(s')$. For any $a \in A$, the notion of \emph{equivalence class} is therefore defined as 
\begin{equation}
    \llbracket a \rrbracket = \{s \in S : l(s) = a\}.
\end{equation}
For any sequence of labels $w = a_1\dots a_k$ of length $k$, this allows to define its probability induced by the Markov chain as 
\begin{equation} \label{eq:prob_induced}
    p^k(w) = \sum_{s_1 \in \llbracket a_1 \rrbracket} \mu_{s_1} \sum_{s_2 \in \llbracket a_2 \rrbracket} \tau_{s_1, s_2} \dots \sum_{s_n \in \llbracket a_n \rrbracket} \tau_{s_{n-1}, s_n}.
\end{equation}

\begin{remark} \label{rem:prob_complexity} 
    Using a similar algorithm as the \emph{forward backward procedure} for hidden Markov models (see \cite{Rabiner1989} and references therein), it is possible to compute $p^l(w)$  for all $w \in A^l$ for all $l = 1, \dots, k$ in $k|S|^2$ operations. \hfill $\square$
\end{remark}

Similarly, we define the finite \emph{behaviour} of the labeled Markov chain $\Gamma$, noted $B(\Gamma) \subseteq A^*$, as the set of all finite sequences $w \in A^k$, for all $k \in \mathbb{N}$, such that $p^k(w) > 0$.

\section{A procedure to construct memory-dependent abstractions} \label{sec:abstractions}

In this section, we first introduce the notion of \emph{adaptive memory abstractions}. We then propose a data-driven procedure to construct tuneable adaptive memory abstractions. 

\subsection{Adaptive-memory abstractions}

We introduce in Definition~\ref{def:abstraction} the notion of \emph{adaptive memory abstraction}, and we give an illustration in Figure~\ref{fig:abstraction}.
\begin{definition}[Adaptive memory abstraction] \label{def:abstraction}
    Given a dynamical system $\Sigma$, a measure $\lambda$ and a set $W$ that defines a partition, an \emph{adaptive memory abstraction} of $\Sigma$ is a labeled Markov chain $\Gamma_{W} = (S, A, \tau, \mu, l)$ composed as follows. 
    \begin{itemize}
        \item The states correspond to the blocks of the partition, that is $S = W$ 
        \item For each node $w_{[-s, t]}$, the initial measure is defined as 
        \begin{equation} \label{eq:initial_MC}
            \mu_{w_{[-s, t]}} = \mathbb{P}_\lambda\left(y_{-s}\dots y_{t} = w\right)
        \end{equation}
        \item For each two nodes $w_{[-s, t]}$ and $w'_{[-s', t']}$ the transition probability is defined as
        \begin{equation} \label{eq:cond_MC}
        \begin{aligned}
            &\tau_{w_{[-s, t]}, w'_{[-s', t']}}\\
            &\quad= \mathbb{P}_{\lambda}\left(
                y_{-s'+1}\dots y_{t'+1} = w'
                \,\middle|\,
               y_{-s}\dots y_{t} = w
            \right) 
        \end{aligned}
        \end{equation}
        \item For each node $w_{[-s, t]} = a_{-s}\dots a_t$, the labeling is defined as $l\left( w_{[-s, t]} \right) = a_0$ \hfill $\square$
    \end{itemize}
\end{definition}

\if\useieeelayout1
\begin{figure}[t]
    \centering
    \includegraphics[width = 0.7\linewidth]{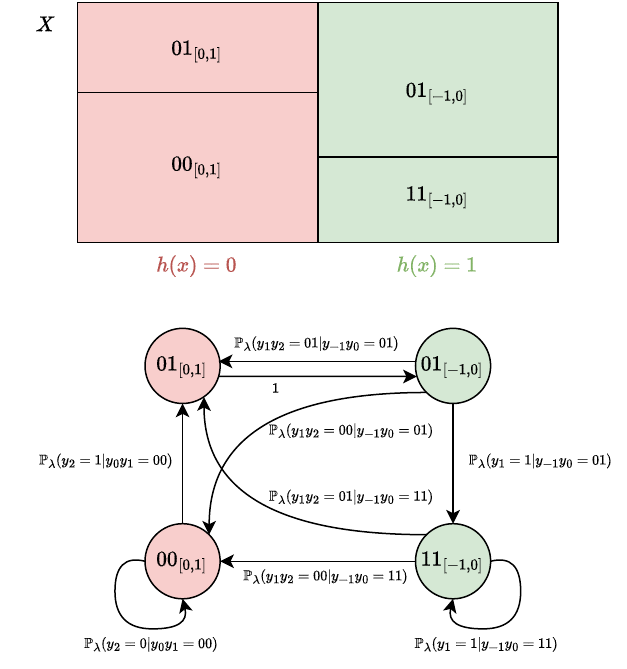}
    \caption{Illustration of an adaptive memory abstraction $\Gamma_W$ for a certain dynamical system $\Sigma$ with measure $\lambda$. In this example, $W = \{01_{[0, 1]}, 00_{[0, 1]}, 01_{[-1, 0]}, 11_{[-1, 0]}\}$, as illustrated above. The corresponding abstraction $\Gamma_W$ is given below, with all the possibly non-zero transition probabilities.} 
    \label{fig:abstraction}
\end{figure}
\else
\begin{figure}[ht]
    \centering
    \includegraphics[width = 0.5\linewidth]{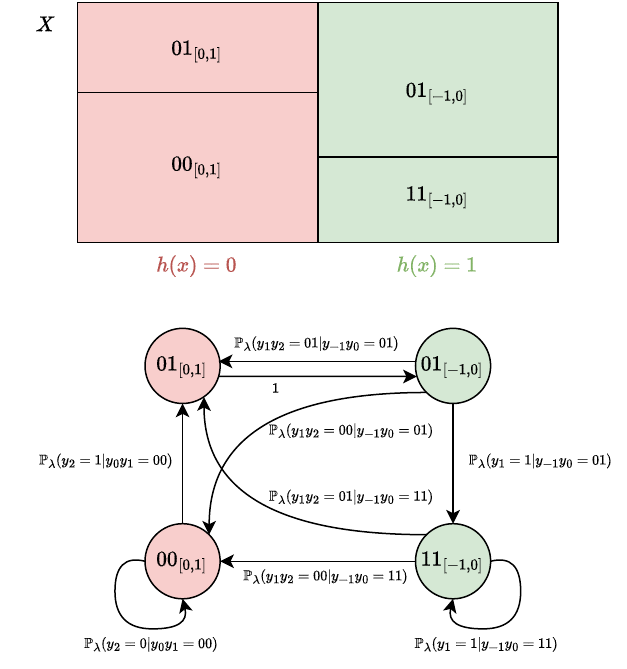}
    \caption{Illustration of an adaptive memory abstraction $\Gamma_W$ for a certain dynamical system $\Sigma$ with measure $\lambda$. In this example, $W = \{01_{[0, 1]}, 00_{[0, 1]}, 01_{[-1, 0]}, 11_{[-1, 0]}\}$, as illustrated above. The corresponding abstraction $\Gamma_W$ is given below, with all the possibly non-zero transition probabilities.} 
    \label{fig:abstraction}
\end{figure}
\fi 

Since $x_0 \sim \lambda$, one needs to make sure that all block $\left[w_{[-s, t]}\right] \subseteq X$ can be captured by the sampling measure.
\begin{assumption} \label{ass:mc_well_defined}
    For all $w_{[-s, t]} \in W$, $\mathbb{P}_{\lambda}(y_{-s}\dots y_t = w) > 0$.
\end{assumption}
Assumption~\ref{ass:mc_well_defined} is a natural necessary condition in the context of data-driven methods. Informally, it states that the behaviour of the dynamical system can be reliably sampled. On a more technical level, it implies that the transition matrix of the corresponding abstraction is stochastic. For the sake of brevity, this result and its proof can be found in Appendix~\ref{app:stochastic_matrix}.


Given a partition $W$, one can sample trajectories and compute the probabilities \eqref{eq:initial_MC} and \eqref{eq:cond_MC} in a Monte-Carlo fashion. For this reason, we say that these abstractions are \emph{data-driven}. Finally, we say that these abstractions are \emph{safe} because, for any partition $W$, all finite sequence output by the dynamical system $\Sigma$ with initial measure $\lambda$ can be simulated by $\Gamma_W$. This is formally written in the following proposition.

\begin{proposition}
    Given a dynamical system $\Sigma$, a measure $\lambda$ and its adaptive memory abstraction $\Gamma_W$, if Assumption~\ref{ass:mc_well_defined} is satisfied, then it holds that $B(\Sigma, \lambda) \subseteq B(\Gamma_W)$.
\end{proposition}
\begin{proof}
    It suffices to show that, for all $k \in \mathbb{N}$, for all words $w_{[0, k-1]} = b_0 \dots b_{k-1}$, if $\lambda([w_{[0, k-1]}]) > 0$ in the original system, then it holds that $p^k(w) > 0$ in the abstraction $\Gamma_W$. Since Assumption~\ref{ass:mc_well_defined} is satisfied, the condition $p^k(w) > 0$ holds if there exists a sequence of states $s_0, \dots, s_{k-1} \in W$ such that $l(s_i) = b_i$ for all $i = 0, \dots, k-1$, and $\lambda([s_i] \cap f^{-1}([s_{i+1}])) > 0$ for all $i = 0, \dots, k-2$. We will therefore prove the latter. 

    First, it holds that 
    \begin{equation}
        \left[w_{[0, k-1]}\right] = \bigcap_{i = 0, \dots, k-1} f^{-i}([(b_i)_{[0, 0]}]), 
    \end{equation}
    where $(b)_{[0, 0]}$ denotes the word composed only of the letter $b$. Moreover, since $W$ defines a partition, it holds that $[(b_i)_{[0, 0]}] = \bigcup_{s_i \in \llbracket b_i \rrbracket} [s_i]$ for all $i = 0, \dots, k-1$. Therefore, one can write
    \begin{equation}
    \begin{aligned}
        \left[w_{[0, k-1]}\right] 
        &= 
        \bigcap_{i = 0, \dots, k-1} f^{-i}\left(
        \bigcup_{s_i \in \llbracket b_i \rrbracket} [s_i]
        \right)\\
        &=
        \bigcap_{i = 0, \dots, k-1}
        \bigcup_{s_i \in \llbracket b_i \rrbracket} f^{-i}([s_i]).
    \end{aligned}
    \end{equation}
    Now, since $\lambda([w_{[0, k-1]}]) > 0$, then there exists at least a sequence $s_0, \dots, s_{k-1} \in W$ such that $l(s_i) = b_i$ and such that 
    $
        \lambda\left( \bigcap_{i = 0, \dots, k-1} f^{-i} ([s_i]) \right) > 0.
    $
    Furthermore, one can say that the inequality above implies that 
    \begin{equation} \label{eq:last_eq_proof}
        \lambda\left(f^{-i}([s_i]) \cap f^{-(i+1)}([s_{i+1}])\right) > 0
    \end{equation}
    for all $i = 0, \dots, k-2$. This is explained by the fact that, for any three sets $A, B, C$, if $\lambda(A \cap B \cap C) > 0$, then it holds that $\lambda(A \cap B) > 0$, $\lambda(A \cap C) > 0$ and $\lambda(B \cap C) > 0$. By invertibility of $f$, \eqref{eq:last_eq_proof} implies that $\lambda\left([s_i] \cap f^{-1}([s_{i+1}])\right) > 0$ for all $i = 0, \dots, k-2$, and the proof is completed.
\end{proof}

\subsection{A data-driven procedure for abstractions}

In the previous section, the choice of a convenient set $W$ is not discussed. In this section, we present a data-driven procedure that tackles this problem. Our procedure, called $\textsc{Refine}$, takes as input a dynamical system $\Sigma$, an initial measure $\lambda$, a \emph{metric} $\mathsf{D}$ between labeled Markov chains, and a number of iterations $N$. It is descripted in Algorithm~\ref{alg:adaptive}, and an illustration of this procedure is given in Figure~\ref{fig:refine}. 

\begin{algorithm}[ht!]
    \caption{$\textsc{Refine}(\Sigma, \lambda, \mathsf{D}, N)$}\label{alg:adaptive}
    \begin{algorithmic}[1]
       \STATE $W \gets \{(a)_{[0, 0]}\}_{a \in A}$ \hfill {\footnotesize Start with a coarse partition}
       \STATE Sample trajectories from $\Sigma$
       \STATE Find $A_0$ s.t. $\mathbb{P}_\lambda(y_{0} = a) = 0$ $\forall a \in A_0$
       \STATE $W \gets W \setminus \{(a)_{[0, 0]}\}_{a \in A_0}$ \label{line:4}
       \STATE Construct $\Gamma_W$ from samples of $\Sigma$ \label{checkpoint}
       \FOR{$n = 1, \dots, N$}
        \FOR{$i = 1, \dots, |W|$}
            \STATE $W_i' \gets W \setminus \left\{w^{i}_{[0, t_i]}\right\}$ \hfill {\footnotesize Try to refine each block}
            \STATE $W_i' \gets W_i' \cup \left\{(w^i a)_{[0, t_i + 1]}\right\}_{a \in A}$ 
            \STATE Sample trajectories from $\Sigma$
            \STATE Find $W_0$ s.t. $\mathbb{P}_\lambda(y_{0}\dots y_t = w) = 0$ $\forall w_{[0, t]} \in W_0$ \label{line:11}
            \STATE $W'_i \gets W'_i \setminus W_0$ \label{line:12}
            \STATE Construct $\Gamma_{W'_{i}}$ from samples of $S$
            \STATE $\mathsf{D}_i \gets \mathsf{D}(\Gamma_W, \Gamma_{W'_i})$
        \ENDFOR
          \STATE $j = \arg\max_{i = 1, \dots, |W|} \mathsf{D}_i$ \hfill {\footnotesize Greedy choice}
          \STATE $W \gets W'_j$
          \STATE $\Gamma_W \gets \Gamma_{W'_j}$
        \ENDFOR
       \RETURN $W$ \hfill {\footnotesize Return a refined partition}
    \end{algorithmic}
\end{algorithm}

\if\useieeelayout1
\begin{figure}[t]
    \centering
    \includegraphics[width = 0.8\linewidth]{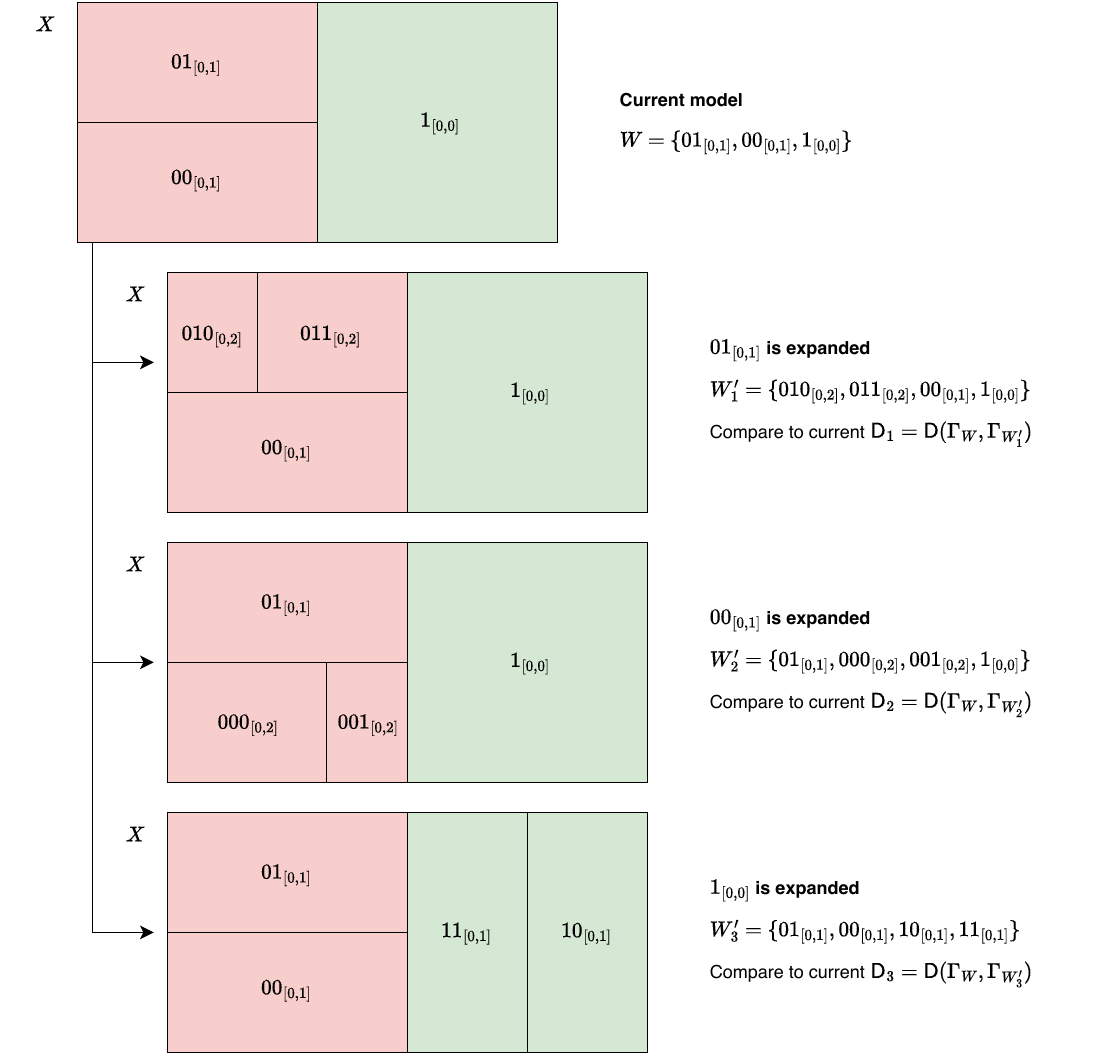}
    \caption{Illustration of the \textsc{Refine} algorithm for a system with $A = \{0, 1\}$. The current model contains three states $W = \left\{01_{[0, 1]}, 00_{[0, 1]}, 1_{[0, 0]}\right\}$. Three sets $W'_1, W'_2, W'_3$ are then constructed, each time expanding one block of the partition into $|A| = 2$ sub-blocks. For each abstraction $\Gamma_{W'_i}$, the distance with the current one is computed: the current model is updated with the model for which the distance is the largest.} 
    \label{fig:refine}
\end{figure}
\else
\begin{figure}[ht!]
    \centering
    \includegraphics[width = 0.6\linewidth]{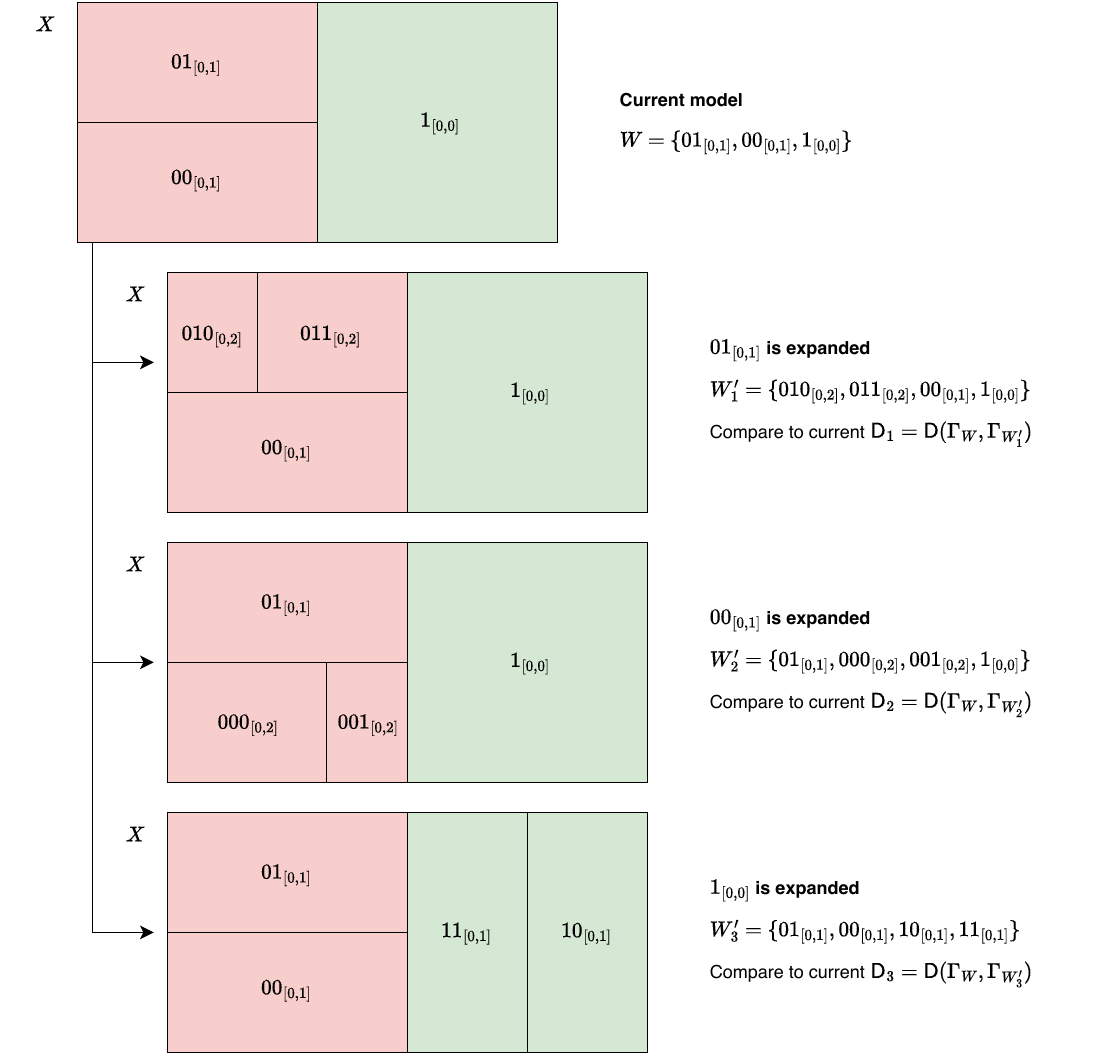}
    \caption{Illustration of the \textsc{Refine} algorithm for a system with $A = \{0, 1\}$. The current model contains three states $W = \left\{01_{[0, 1]}, 00_{[0, 1]}, 1_{[0, 0]}\right\}$. Three sets $W'_1, W'_2, W'_3$ are then constructed, each time expanding one block of the partition into $|A| = 2$ sub-blocks. For each abstraction $\Gamma_{W'_i}$, the distance with the current is computed. The current model is updated with the model for which the distance is the largest.} 
    \label{fig:refine}
\end{figure}
\fi 

The algorithmic procedure $\textsc{Refine}$ presented here depends on  a general notion of metric between Markov models, rather than a particular one as in \cite[Algorithm~2]{Banse2023}. As it is based solely on trajectories of the original dynamics, we refer to $\textsc{Refine}$ as a data-driven abstraction algorithm. Besides, as formally stated in Appendix~\ref{app:procedure}, under the assumption that any output trace can be sampled with non-zero measure, $\textsc{Refine}$ leads to a valid partition of the state space of the original dynamics.

\begin{remark}
    Many variants of the \textsc{Refine} algorithm can be considered. For example, one could expand the memory in the past by expanding each word $w^i_{[-s_i, 0]}$ into $\{(aw^i)_{[-(s_i + 1), 0]}\}_{a \in A}$. Also, inspired by algorithms from reinforcement learning (cf. TD learning scheme in \cite{Sutton}), one could compare the current model with the $|A|^n$ models possible after $n$ steps instead of choosing between the possible $|A|$ different models, and take the one for which the distance is the largest. \hfill $\square$
\end{remark}

The output $W$ from the \textsc{Refine} algorithm depends on the chosen metric between the abstractions. In the following section, we introduce a novel metric between labeled Markov chains, and we discuss the interpretation of the corresponding \textsc{Refine} output.

\section{The Cantor-Kantorovich metric} \label{sec:ck-distance}

In this section, we introduce a new metric $\mathsf{CK}(\Gamma_1, \Gamma_2)$, named Cantor-Kantorovich metric, between two labeled Markov chains. This new metric was first introduced in \cite{Banse2023}, but without proofs. Here we prove that the metric is well-defined\footnote{In the sense that it satisfies positivity, symmetry and triangle inequality.}, then we present an algorithm to approximate it efficiently. 

\subsection{Definition}

Given any two Markov chains $\Gamma_1 = (S_1, A, \tau_1, \mu_1, l_1)$ and $\Gamma_2 = (S_2, A, \tau_2, \mu_2, l_2)$ defined on the same set of outputs, for a fixed $k \in \mathbb{N}_{> 0}$, let $p_1^k : A^k \to [0, 1]$ and $p_2^k: A^k \to [0, 1]$ be the probabilities induced by $\Gamma_1$ and $\Gamma_2$ such as defined in \eqref{eq:prob_induced}.

\begin{definition}[Kantorovich metric] 
    Let $\mathsf{D} : A^k \times A^k \to \mathbb{R}_{> 0}$ be any metric between words of length $k$. The \emph{Kantorovich metric} between $p_1^k$ and $p_2^k$ is defined as
    \begin{equation} \label{eq:kant_LP}
        \mathsf{K}_{\mathsf{D}}(p_1^k, p_2^k)
        = \min_{\pi^k \in \Pi(p_1^k, p_2^k)} \sum_{w_1, w_2 \in A^k} \mathsf{D}(w_1, w_2)\pi^k(w_1, w_2), 
    \end{equation}
    where $\Pi(p_1^k, p_2^k)$ is the set of all \emph{couplings} of $p_1^k$ and $p_2^k$, that is the set of all joint distribution $\pi^k : A^k \times A^k \to [0, 1]$ such that the constraints
    \begin{equation} \label{eq:cons_LP_pos}
        \pi^k(w_1, w_2) \geq 0 \quad \forall w_1, w_2 \in A^k,
    \end{equation}
    \begin{equation}
    \begin{gathered} \label{eq:cons_LP_marg}
        \sum_{w_2 \in A^k} \pi^k(w_1, w_2) = p_1^k(w_1) \quad \forall w_1 \in A^k, \\
        \sum_{w_1 \in A^k} \pi^k(w_1, w_2) = p_2^k(w_2) \quad \forall w_2 \in A^k
    \end{gathered}
    \end{equation}
    hold. \hfill $\square$
\end{definition}

For a given $k$, the definition of the Kantorovich metric depends on an underlying distance $\mathsf{D}$ over the set of words of length $k$. We propose to endow the latter with the \emph{Cantor distance}, defined as follows (see e.g. \cite[Section~2.1]{Rozenberg2002}). 
\begin{definition}[Cantor distance]
    The \emph{Cantor distance} between any two sequences $w_1 = a_1\dots a_k$ and $w_2 = b_1\dots b_k$ is defined as 
    \begin{equation}
        \mathsf{C}(w_1, w_2)
        = \inf\{ 2^{-|c|} : \text{$c$ is a common prefix of $w_1$ and $w_2$}\}
    \end{equation}
    if $w_1 \neq w_2$, and $\mathsf{C}(w_1, w_2) = 0$ otherwise. \hfill $\square$
\end{definition}

\begin{lemma}[See \cite{Rozenberg2002}] \label{lemma:ultra}
    The Cantor distance satisfies the strong triangular inequality. That is, for any $w_1, w_2, w_3 \in A^k$, 
    \begin{equation}
        \mathsf{C}(w_1, w_3) \leq \max\{
            \mathsf{C}(w_1, w_2), \mathsf{C}(w_2, w_3)
        \}.
    \end{equation}
\end{lemma}

We now define the Cantor-Kantorovich metric between two labeled Markov chains. 
\begin{definition}[Cantor-Kantorovich metric]
    Let $\Gamma_1$ and $\Gamma_2$ be two labeled Markov chains defined on the same set of labels $A$. Their \emph{Cantor-Kantorovich metric} is defined as 
    \begin{equation}
        \mathsf{CK}(\Gamma_1, \Gamma_2) = \lim_{k \to \infty} \mathsf{K}_\mathsf{C}(p_1^k, p_2^k), 
    \end{equation}
    where $\mathsf{K}_\mathsf{C}$ is the Kantorovich metric with the Cantor distance as underlying distance, and where $p_1^k$ and $p_2^k$ are the probabilities respectively induced by $\Gamma_1$ and $\Gamma_2$. \hfill $\square$
\end{definition}
We prove in the next section that this metric is well-defined, in the sense that it satisfies positivity, symmetry and triangle inequality.

For a fixed length $k$, computing the Kantorovich metric $\mathsf{K}_\mathsf{C}(p_1^k, p_2^k)$ can be achieved by solving the linear program \eqref{eq:kant_LP}. The latter can be seen as an optimal transport problem where, at each $w \in A^k$, the ``supplies'' are given by $p_1^k(w)$ and the ``demands'' by $p_2^k(w)$. For each $w_1, w_2 \in A^k$, the cost of moving $\pi^k(w_1, w_2)$ of supply mass from $w_1$ to $w_2$ amounts to $\pi^k(w_1, w_2)\mathsf{C}(w_1, w_2)$. The specific choice of the Cantor distance as the underlying distance for the metric, allows to visualise this with a tree, as shown in Figure~\ref{fig:cant_kant_tree}.

\if\useieeelayout1
\begin{figure}[ht!]
    \centering
    \includegraphics[width = 0.45\linewidth]{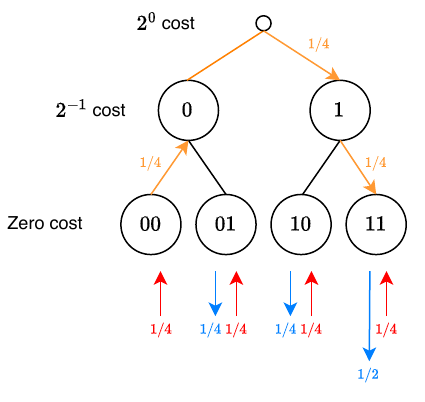}
    \caption{Illustration of the Kantorovich metric linear program \eqref{eq:kant_LP} for $A = \{0, 1\}$ and $k = 2$. The supplies are given in red, and the demands in blue. To solve this problem, one has to move $\pi^2(00, 11) = 1/4$ supply mass from 00 to 11. To do that, the mass has to travel up to the root. The corresponding Cantor distance is $\mathsf{C}(00, 11) = 2^{0}$. The total cost is therefore $\mathsf{K}_\mathsf{C}(p_1^2, p_2^2) = 1/4$.}
    \label{fig:cant_kant_tree}
\end{figure}
\else 
\begin{figure}[ht!]
    \centering
    \includegraphics[width = 0.4\linewidth]{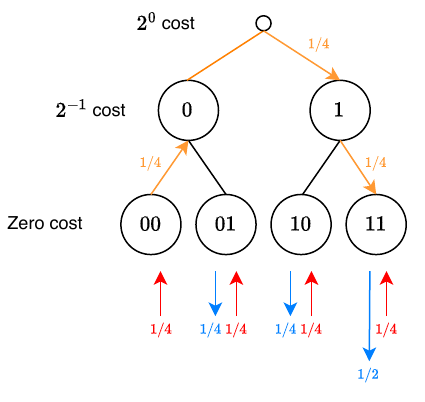}
    \caption{Illustration of the Kantorovich metric linear program \eqref{eq:kant_LP} for $A = \{0, 1\}$ and $k = 2$. The supplies are given in red, and the demands in blue. To solve this problem, one has to move $\pi^2(00, 11) = 1/4$ supply mass from 00 to 11. To do that, the mass has to travel up to the root. The corresponding Cantor distance is $C(00, 11) = 2^{0}$. The total cost is therefore $K(p_1^2, p_2^2) = 1/4$.}
    \label{fig:cant_kant_tree}
\end{figure}
\fi 

\begin{remark} \label{rem:lp}
    A na\"ive approach to compute the metric  is to use linear programming (LP) or combinatorial optimisation (CO) methods \cite{Oberman2015}. However, such methods are simply too costly to be used in practice. In the best case, the complexity is worse than quadratic, that is the number of operations is $\Omega(|A|^{2k})$. In the following section, we leverage the particular underlying Cantor distance to derive an algorithm that computes $\mathsf{K}_\mathsf{C}(p_1^k, p_2^k)$ that scales better. \hfill $\square$
\end{remark}

The Cantor-Kantorovich metric defines a metric space in which two labeled Markov chains are close if they have similar short-horizon behaviours. 
Indeed, the Cantor distance can be interpreted as a discount factor, and a large Cantor-Kantorovich metric means that the probabilities on sequences of labels differ close to the initial step of the random walks. The procedure $\textsc{Refine}(\Sigma, \lambda, \mathsf{CK}, N)$ therefore tends to choose the model with a different short-term behaviour. 

\subsection{A recursive algorithm for approximating the Cantor-Kantorovich metric}

In this section, we state Theorem~\ref{thm:recursive}, a central recursive result for computing the Cantor-Kantorovich metric. This fact will be useful for two things. First, it implies that the Cantor-Kantorovich metric satisfies positivity, symmetry and triangle inequality, and that it can be approximated, as stated in Theorem~\ref{thm:well-defined}. Second, it provides an efficient algorithm to approximate it.

We first state two lemmata that will be useful to prove Theorem~\ref{thm:recursive}. The proofs of the latter are moved in the appendices (see Appendix~\ref{app:proof_firstlemma} and Appendix~\ref{app:proof_secondlemma}). 
\begin{lemma} \label{lemma:firstlemma}
    For any $k \geq 1$, let $\pi^k$ be the solution of \eqref{eq:kant_LP}. For all $w \in A^k$, it holds that $\pi^k(w, w) = \min\{p_1^k(w), p_2^k(w)\}$.
\end{lemma}
\begin{lemma} \label{lemma:secondlemma}
    For any $k \geq 2$, let $\pi^k$ be the solution of \eqref{eq:kant_LP}. For all $w \in A^{k-1}$ such that $p_1^{k-1}(w) > p_2^{k-1}(w)$, then 
    \begin{equation}
    \begin{aligned}
        \sum_{\substack{w' \in A^{k-1} \\ w' \neq w}} \sum_{a_1, a_2 \in A} \pi^k(wa_1, w'a_2) &= p_1^{k-1}(w) - p_2^{k-1}(w) \\
        \sum_{\substack{w' \in A^{k-1} \\ w' \neq w}} \sum_{a_1, a_2 \in A} \pi^k(w'a_1, wa_2) &= 0.
    \end{aligned}
    \end{equation}
    Else if $p_1^{k-1}(w) \leq p_2^{k-1}(w)$, then 
    \begin{equation}
        \begin{aligned}
            \sum_{\substack{w' \in A^{k-1} \\ w' \neq w}} \sum_{a_1, a_2 \in A} \pi^k(wa_1, w'a_2) &= 0\\
            \sum_{\substack{w' \in A^{k-1} \\ w' \neq w}} \sum_{a_1, a_2 \in A} \pi^k(w'a_1, wa_2) &= p_2^{k-1}(w) - p_1^{k-1}(w).
        \end{aligned}
        \end{equation}
\end{lemma}

\begin{theorem} \label{thm:recursive}
    For any $k \geq 1$, let $\pi^k$ be the solution of \eqref{eq:kant_LP}. Then it holds that
    \begin{equation} \label{eq:iterative_prop}
    \begin{aligned}
        &\mathsf{K}_\mathsf{C}(p_1^{k+1}, p_2^{k+1}) = \mathsf{K}_\mathsf{C}(p_1^k, p_2^k) \\
        &\quad + 2^{-k} 
        \sum_{w \in A^k} \left( r^k(w) - \sum_{a \in A} r^{k+1}(wa) \right), 
    \end{aligned}
    \end{equation}
    where $r^k(w) = \min\{p_1^k(w), p_2^k(w)\}$ for any $w \in A^k$.
\end{theorem}

\begin{proof}
For the sake of conciseness, we note $\mathsf{K}_{\mathsf{C}}^k = \mathsf{K}_{\mathsf{C}}(p_1^k, p_2^k)$. This proof is divided into two parts. First, we prove that the right-hand side of \eqref{eq:iterative_prop} is a lower bound for $\mathsf{K}_{\mathsf{C}}^{k+1}$. Second, we prove that it is also an upper bound.

First we note that $\mathsf{K}_{\mathsf{C}}^{k+1}$ is equal to
\begin{equation}
\begin{aligned}
    &\sum_{w_1, w_2 \in A^k} \sum_{a_1, a_2 \in A} \mathsf{C}(w_1a_1, w_2a_2) \pi^{k+1}(w_1a_1, w_2a_2) \\
    =& \sum_{\substack{w_1, w_2 \in A^k\\w_1 \neq w_2}} \mathsf{C}(w_1, w_2) \sum_{a_1, a_2 \in A} \pi^{k+1}(w_1a_1, w_2a_2) \\
    &+ 2^{-(k+1)} \sum_{w \in A^k} \sum_{\substack{a_1, a_2 \in A\\a_1 \neq a_2}} \pi^{k+1}(wa_1, wa_2) \\
    &:= C_1 + C_2.
\end{aligned}
\end{equation}
In the first part of this proof, we show that the two following expressions hold:
\begin{equation} \label{eq:first_sketch}
    C_1 \geq \mathsf{K}_{\mathsf{C}}^k, 
\end{equation}
\begin{equation} \label{eq:second_sketch}
    C_2 = 2^{-k} \sum_{w \in A^k} \left( r^k(w) - \sum_{a \in A} r^{k+1}(wa) \right).
\end{equation}
In the first instance, we show that \eqref{eq:first_sketch} holds. To do this, let $\mu^k : A^k \times A^k \to [0, 1]$ be defined as 
$
    \mu^k(w_1, w_2) = \sum_{a_1, a_2 \in A} \pi^{k+1}(w_1a_1, w_2a_2).
$
We show that $\mu^k$ satisfies the constraints \eqref{eq:cons_LP_pos} and \eqref{eq:cons_LP_marg}. Indeed, $\mu^k(w_1, w_2) \geq 0$, and
\begin{equation}
\begin{aligned}
    \sum_{w_2 \in A^k} \mu(w_1, w_2) &= \sum_{w_2 \in A^k} \sum_{a_1, a_2 \in A} \pi^{k+1}(w_1a_1, w_2a_2)\\
    &= \sum_{a_1} p_1^{k+1}(w_1a_1) = p_1^k(w_1), 
\end{aligned}
\end{equation}
and similarly for the second condition in \eqref{eq:cons_LP_marg}. This implies that $\mu^k$ is a coupling, thereby a feasible solution of \eqref{eq:kant_LP}. This yields 
\begin{equation}
    \mathsf{K}_{\mathsf{C}}^k \leq \sum_{a_1, a_2 \in A} \mathsf{C}(w_1, w_2) \mu^k(w_1, w_2) = C_1.
\end{equation}
As a second step, we show that \eqref{eq:second_sketch} holds. More precisely, we show that, for all $w \in A^k$, the following holds:
\begin{equation}
    \sum_{\substack{a_1, a_2 \in A\\a_1 \neq a_2}} \pi^{k+1}(wa_1, wa_2) = r^k(w) - \sum_{a \in A} r^{k+1}(wa), 
\end{equation}
which implies that 
\begin{equation}
    C_2 = 2^{-(k+1)} \sum_{w \in A^k} \left[ r^k(w) - \sum_{a \in A} r^{k+1}(wa) \right].
\end{equation}
We prove the claim. Assume without loss of generality that $w$ is such that $p_1^k(w) > p_2^k(w)$, then 
\begin{equation}
\begin{aligned}
    &\sum_{a_1, a_2 \in A} \pi^{k+1}(wa_1, wa_2) \\
    & = \sum_{a_1 \in A} 
        \sum_{w' \in A^k} \sum_{a_2 \in A} 
        \pi^{k+1}(wa_1, w'a_2) \\
    & \quad \quad \quad \quad - \sum_{\substack{w' \in A^k\\w' \neq w}}
    \sum_{a_2 \in A} 
    \pi^{k+1}(wa_1, w'a_2) \\
    &= \sum_{a_1 \in A} p_1^k(wa_1) - \sum_{\substack{w' \in A^k\\w' \neq w}} \sum_{a_1, a_2 \in A} 
    \pi^{k+1}(wa_1, w'a_2) \\
\end{aligned}  
\end{equation}
Following Lemma~\ref{lemma:secondlemma}, this is equal to 
$
    p_1^k(w) - \left(p_1^k(w) - p_2^k(w)\right) = r^k(w).
$
And the following holds: 
\begin{equation}
\begin{aligned}
    &\sum_{\substack{a_1, a_2 \in A\\a_1 \neq a_2}} \pi^{k+1}(wa_1, wa_2) \\
    =& \sum_{\substack{a_1 \neq a_2 \in A\\a_1 \neq a_2}} \pi^{k+1}(wa_1, wa_2) - \sum_{a \in A} \pi^{k+1}(wa_1, wa_2) \\
    =&\, r^k(w) - \sum_{a \in A} r^{k+1}(wa)
\end{aligned}
\end{equation}
by Lemma~\ref{lemma:firstlemma}. This concludes that the right-hand side of \eqref{eq:iterative_prop} is a lower bound for $\mathsf{K}_{\mathsf{C}}^{k+1}$. 

We now move to the second part of this proof. To provide an upper bound, we will show that we can construct a feasible $k+1$ solution feasible $\mu^{k+1}$ such that 
\begin{equation}
\begin{aligned}
    &\sum_{w_1, w_2 \in A^{k+1}} \mathsf{C}(w_1, w_2)\mu^{k+1}(w_1, w_2) \\
    &\quad \quad \quad = \mathsf{K}_{\mathsf{C}}^k + \sum_{w \in A^k} \left[
        r^k(w) - \sum_{a \in A} r^{k+1}(wa)
    \right].
\end{aligned}
\end{equation}
Consider $\pi^k$, an optimal solution at step $k$. We will construct $\mu^{k+1}$ in the following greedy way. Initialise $\mu^{k+1}$ with only zero elements, and for all $w \in A^{k}$, $a \in A$, we initialise $\delta(wa) = 0$. We start by updating the blocks $\mu^{k+1}(w_1a_1, w_2a_2)$ where $w_1 \neq w_2$. For all $w$ such that $p_1^{k}(w) > p_2^{k}(w)$, for all $a \in A$ such that $p_1^{k+1}(wa) > p_2^{k+1}(wa)$, do the following.
\begin{enumerate}
\item Let $\tilde{\delta}(wa) = p_1^{k+1}(wa) - p_2^{k+1}(wa)$.\\
If $\sum_{a' \neq a} \delta(wa') + \tilde{\delta}(wa) > p_1^{k}(w) - p_2^{k}(w)$, let $\delta(wa) = (p_1^{k}(w) - p_2^{k}(w)) - \sum_{a' \neq a} \delta(wa')$.\\
Else let $\delta(wa) = \tilde{\delta}(wa)$.
\item \label{item:current} Find a $w' \neq w$ such that
\begin{equation}
    \pi^k(w, w') > \sum_{a_1, a_2 \in A} \mu^{k+1}(wa_1, w'a_2). 
\end{equation}
Now, for any $a' \in A$, let 
\begin{equation}
\begin{aligned}
    &\psi(a') = p_2^{k+1}(w'a') - p_1^{k+1}(w'a')  \\
    & - \sum_{\substack{w'' \in A^k\\w'' \neq w'}} \sum_{a_1 \in A} \mu^{k+1}(w''a_1, w'a')
\end{aligned}
\end{equation}
Then, find $a' \in A$ such that $\psi(a') > 0$. \\
Now, if $\delta(wa) > \psi(a')$, then: 
\begin{itemize}
    \item Update $\mu(wa, w'a') \gets \psi(a')$
    \item Update $\delta(wa) \gets \delta(wa) - \psi(a')$
    \item Return to \ref{item:current}.
\end{itemize}
Else, update $\mu(wa, w'a') \gets \delta(wa)$.
\end{enumerate}
We claim that, in the procedure above, there always exists such a $w'$ for a given $wa$. Otherwise, 
$
    \sum_{w' \neq w}\pi^k(w, w') < p_1^k(w) - p_2^k(w),  
$
which is impossible by Lemma~\ref{lemma:firstlemma}. Also, we claim that there also always exists such $a'$ for a given $wa$ and $w'$. Otherwise, for all $a' \in A$, 
\begin{equation}
\begin{aligned}
    &\sum_{a' \in A} \sum_{\substack{w'' \in A^k\\w'' \neq w'}} \sum_{a_1 \in A} \mu^{k+1}(w''a_1, w'a') \\
    &\quad = \sum_{a' \in A} p_2^{k+1}(w'a') - p_1^{k+1}(w'a'), 
\end{aligned}
\end{equation}
which means by construction that 
\begin{equation}
\begin{aligned}
    &\sum_{a' \in A} \sum_{\substack{w'' \in A^k\\w'' \neq w'}} \sum_{a_1 \in A} \pi^{k+1}(w''a_1, w'a') \\
    &\quad > \sum_{a' \in A} p_2^{k+1}(w'a') - p_1^{k+1}(w'a'), 
\end{aligned}
\end{equation}
which is $p_2^{k}(w) - p_1^{k}(w) > p_2^{k}(w) - p_1^{k}(w)$ by Lemma~\ref{lemma:secondlemma}. Moreover, by construction we have that, for all $w \neq w'$,
\begin{equation} \label{eq:non_diag_blocks}
    \pi^{k}(w, w') = \sum_{a_1, a_2 \in A^k} \mu^{k+1}(wa_1, w'a_2).
\end{equation}

Now, we construct the diagonal blocks $\mu^{k+1}(wa_1, wa_2)$. For each $w \in A^k$ and $a \in A$, let 
\begin{equation} \label{eq:fake_constraints}
\begin{gathered}
    \tilde{p}_1^{k+1}(wa) = p_1^{k+1}(wa) - \sum_{w' \neq w} \sum_{a \in A} \mu^{k+1}(wa, w'a'),\\
    \tilde{p}_2^{k+1}(wa) = {p}_2^{k+1}(wa) - \sum_{w' \neq w} \sum_{a \in A} \mu^{k+1}(w'a', wa).
\end{gathered}
\end{equation}
Now, for a given $w$, let us solve the following balanced optimal transport problem: 
\begin{equation}
\begin{aligned}
    &\inf_{\mu^{k+1}} 2^{-k}\sum_{\substack{a_1, a_2 \in A\\a_1 \neq a_2}} \mu^{k+1}(wa_1, wa_2) \\
    \textrm{s.t. }&\forall a_1 \in A: \, \sum_{a_2} \mu^{k+1}(wa_1, wa_2) = \tilde{p}_1^{k+1}(wa_1), \\
    &\forall a_2 \in A: \, \sum_{a_1} \mu^{k+1}(wa_1, wa_2) = \tilde{p}_2^{k+1}(wa_2).
\end{aligned}
\end{equation}
Following the definition of $\tilde{p}$ and $\tilde{q}$, this is a balanced optimal transport whose trivial solution is given by 
\begin{equation} \label{eq:trivial_balanced}
    2^{-k} \left(r^k(w) - \sum_{a \in A} r^{k+1}(wa)\right).
\end{equation}

Now we conclude the proof. By \eqref{eq:non_diag_blocks} and \eqref{eq:fake_constraints}, $\mu^{k+1}$ is a coupling of $p_1^{k+1}$ and $p_2^{k+1}$. Indeed it is positive, and for any $w_1 \in A^k$ and $a_1 \in A$, 
\begin{equation}
\begin{aligned}
    &\sum_{w_2 \in A^k} \sum_{a_2 \in A} \mu^{k+1}(w_1a_1, w_2a_2) \\
    = &\sum_{a_2 \in A} \mu^{k+1}(w_1a_1, w_1a_2) + \sum_{\substack{w_2 \in A^k\\w_2 \neq w_1}} \sum_{a_2 \in A} \mu^{k+1}(w_1a_1, w_2a_2) \\
    = \, &\tilde{p}_1^{k+1}(w_1a_1) + \left(p_1^{k+1}(w_1a_1) - \tilde{p}_1^{k+1}(w_1a_1)\right)= p_1^{k+1}(w_1a_1), 
\end{aligned}
\end{equation}
and similarly for $p_2^{k+1}$. Finally, 
\begin{equation}
\begin{aligned}
    &\sum_{w_1, w_2 \in A^k} \sum_{a_1, a_2 \in A} \mathsf{C}(w_1a_1, w_2a_2) \mu^{k+1}(w_1a_1, w_2a_2) \\
    =& \sum_{w_1 \neq w_2} \mathsf{C}(w_1, w_2) \sum_{a_1, a_2} \mu^{k+1}(w_1a_1, w_2a_2) \\
    +& 2^{-k} \sum_{w} \sum_{\substack{{a_1, a_2}\\a_1 \neq a_2}} \mu^{k+1}(wa_1, wa_2).
\end{aligned}
\end{equation}
By \eqref{eq:non_diag_blocks}, the first term is $\mathsf{K}_{\mathsf{C}}^k$, and by \eqref{eq:trivial_balanced}, the second term is 
$
    \sum_{w \in A^k} \left[r^k(w) - \sum_{a \in A} r^{k+1}(wa) \right].
$
This provides an upper bound on $\mathsf{K}_{\mathsf{C}}^{k+1}$, and the proof is completed. 
\end{proof}



Theorem~\ref{thm:recursive} is central for proving Theorem~\ref{thm:well-defined}, stated next. The latter first states that $\mathsf{CK}$ satisfies positivity, symmetry and triangle inequality. This also implies that the Cantor-Kantorovich metric can be approximated with $\mathsf{K}_\mathsf{C}(p_1^k, p_2^k)$ for a finite $k$. 
\begin{theorem} \label{thm:well-defined}
    The function $\mathsf{CK}$ satisfies positivity, symmetry and triangle inequality. Moreover, for any $k \geq 1$, 
    \begin{equation} \label{eq:approximate}
        0 \leq \mathsf{CK}(\Gamma_1, \Gamma_2) - \mathsf{K}_{\mathsf{C}}(p_1^k, p_2^k) \leq 2^{1-k} \sum_{w \in A^k} r^{k}(w). 
    \end{equation}
\end{theorem}
\begin{proof}
    For the sake of conciseness, let $\mathsf{K}_{\mathsf{C}}^k := \mathsf{K}_{\mathsf{C}}(p_1^k, p_2^k)$, and let $S_k := \sum_{w \in A^k} r^{k}(w)$. Proving that the metric is well-defined reduces to proving that the sequence $(\mathsf{K}_{\mathsf{C}}^k)_{k \geq 1}$ converges. Theorem~\ref{thm:recursive} implies that 
    \begin{equation} \label{eq:KC_exp}
        \mathsf{K}_{\mathsf{C}}^k = (1 - S_1) + \sum_{i = 1}^{k - 1} 2^{-i}(S_i - S_{i+1}).
    \end{equation}
    
    Now let us focus on $S_i - S_{i+1}$. By the law of total probability, for all $i \geq 1$, we have that 
    $
        0 \leq r^i(w) - \sum_{a \in A} r^{i+1}(wa) \leq r^i(w), 
    $
    and therefore, by summing over all the words $w \in A^i$, we have that
    \begin{equation} \label{eq:Si-Si+1}
        0 \leq S_i - S_{i+1} \leq S_i \leq 1.
    \end{equation}
    
    One can see that the maximal value of \eqref{eq:KC_exp} is attained when $S_1 = 0$, which implies that $S_i = 0$ for $i = 2, 3, \dots$, and therefore the sequence $(\mathsf{K}_{\mathsf{C}}^k)_{k \geq 1}$ is upper bounded by $1$. Moreover, one can assert that the sequence is also monotone, following \eqref{eq:Si-Si+1}. By the monotone convergence theorem, the limit therefore exists and is equal to $
        \mathsf{CK}(\Gamma_1, \Gamma_2) = \lim_{k \to \infty} \mathsf{K}_{\mathsf{C}}^k = \sup_{k \geq 1} \mathsf{K}_{\mathsf{C}}^k.$

    Now, it remains to show that \eqref{eq:approximate} holds. Inequalities \eqref{eq:Si-Si+1} implies that
    \begin{equation}
    \begin{aligned}
        &\mathsf{CK}(\Gamma_1, \Gamma_2) - \mathsf{K}_{\mathsf{C}}^p 
        = \sum_{i = p}^\infty 2^{-i} (S_i - S_{i+1}) \\
        &\quad \quad \leq \sum_{i = p}^\infty 2^{-i} S_i \leq S_p \sum_{i = p}^\infty 2^{-i} = 2^{1-p} S_p, 
    \end{aligned} 
    \end{equation}
    which is the claim.
\end{proof}

\begin{remark} \label{rem:approximation}
    Following Theorem~\ref{thm:well-defined}, one can compute a-priori the number of iterations needed to reach an accuracy $\varepsilon \in (0, 1)$. Indeed the right-hand side of \eqref{eq:approximate} is upper bounded by $2^{1-k}$, and therefore it suffices to take $k = \lceil \log_2(\varepsilon^{-1}) \rceil + 1$ to guarantee an $\varepsilon$-accurate solution. \hfill $\square$
\end{remark}

Based on the above results, let us introduce Algorithm~\ref{alg:kantorovich}, which takes as input two labeled Markov chains $\Gamma_1$ and $\Gamma_2$, and a desired accuracy $\varepsilon$, and computes the metric of interest. This algorithm relies on the recursive algorithm \textsc{CK-rec}, which is a natural implementation of Theorem~\ref{thm:recursive} that is described in Algorithm~\ref{alg:kantorovich_rec}. 

\begin{algorithm}[ht!]
    \caption{$\textsc{CK-rec}(\textsc{acc}, l, w, k)$}\label{alg:kantorovich_rec}
    \begin{algorithmic}[1]
        \STATE $r = \min\{ p_1^l(w), p_2^l(w) \} $
        \IF{$r = 0$}
            \STATE Stop
        \ENDIF
        \IF{$l = k$}
            \STATE $\textsc{acc} \gets \textsc{acc} + 2^{1-k}r$
            \STATE Stop
        \ENDIF
        \STATE $\textsc{acc} \gets \textsc{acc} + 2^{-k}r$
        \FOR{$a \in A$}
            \STATE $\textsc{CK-rec}(\textsc{acc}, l+1, wa, k)$
        \ENDFOR
    \end{algorithmic}
 \end{algorithm}

\begin{algorithm}[ht!]
    \caption{$\textsc{Cantor-Kantorovich}(\Gamma_1, \Gamma, \varepsilon)$}\label{alg:kantorovich}
    \begin{algorithmic}[1]
        \STATE For $l = 1, \dots, k$, compute $p_1^l(w)$ and $p_2^l(w)$ \hfill {\footnotesize See Remark~\ref{rem:prob_complexity}}
        \STATE $k \gets \lceil \log_2(\varepsilon^{-1}) \rceil + 1$ \hfill {\footnotesize See Remark~\ref{rem:approximation}}
        \STATE $\textsc{acc} \gets 0$
        \FOR{$a \in A$}
            \STATE $\textsc{CK-rec}(\textsc{acc}, 1, a, k)$
        \ENDFOR
        \RETURN $\textsc{acc}$
    \end{algorithmic}
\end{algorithm}

Algorithm~\ref{alg:kantorovich} terminates in $2k|S|^2 + \mathcal{O}(|A|^{k+1})$ operations. Indeed, the first term is the number of operations needed to compute the probabilities $p_1^l(w)$ and $p_2^l(w)$ for $l = 1, \dots, k$ (see Remark~\ref{rem:prob_complexity}), and the second term is the number of operations corresponding for a DFS\footnote{Depth-First-Search, see e.g. \cite{Sedgewick2011} for an introduction.} in a tree of $|A|^{k+1}$ nodes with a constant number of operations at each node. In order to illustrate the gain in complexity compared to classical LP or CO methods, we provide in Figure~\ref{fig:complexity} the functions $|A|^{k+1}$ and $|A|^{2k}$ for $|A| = 2, 3, 4$ for $k = 1, \dots, 15$ (see Remark~\ref{rem:lp}). 


\if\useieeelayout1
\begin{figure}[t]
    \centering
    \includegraphics[width = 0.7\linewidth]{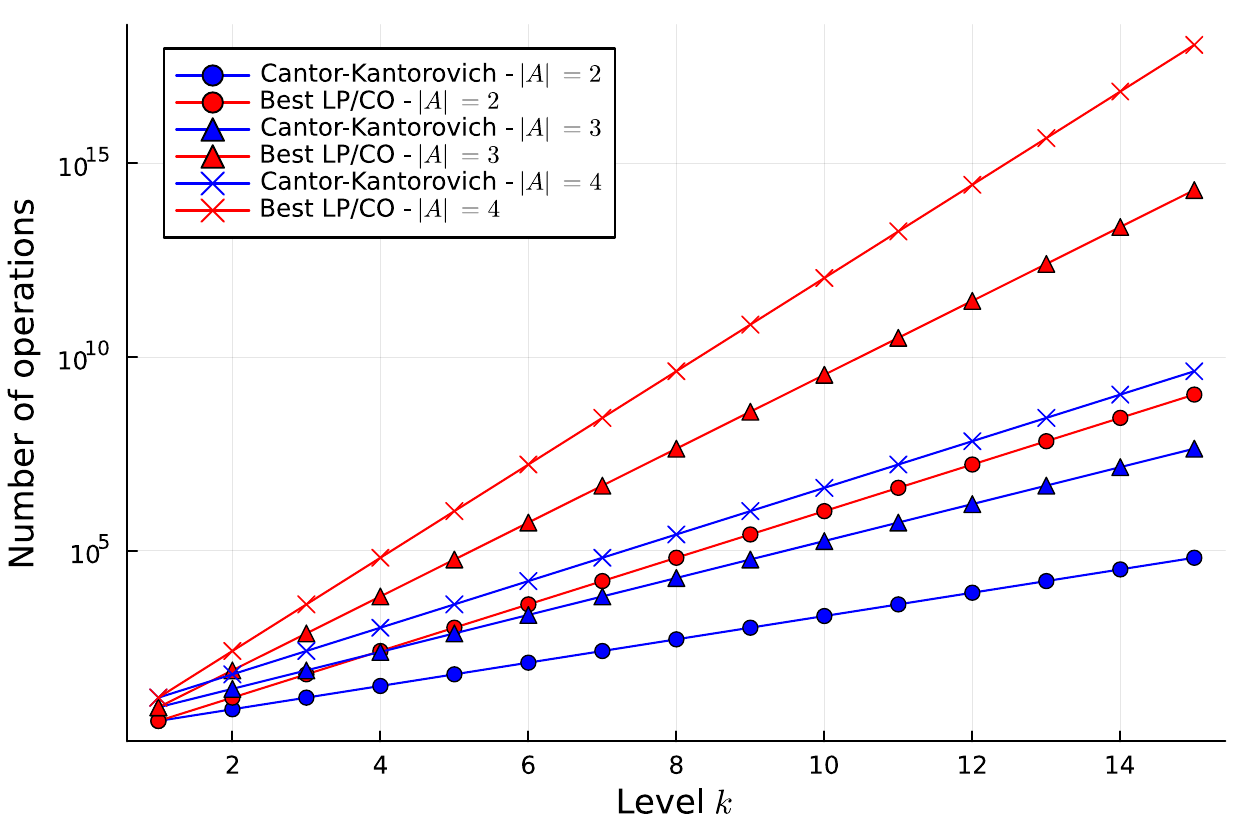}
    \caption{Comparaison of the computational complexity of two methods to solve \eqref{eq:kant_LP} for $|A| = 2, 3, 4$. In red, the best LP/CO method, and in blue our method as described in Algorithm~\ref{alg:kantorovich}.}
    \label{fig:complexity}
\end{figure}
\else 
\begin{figure}[ht!]
    \centering
    \includegraphics[width = 0.5\linewidth]{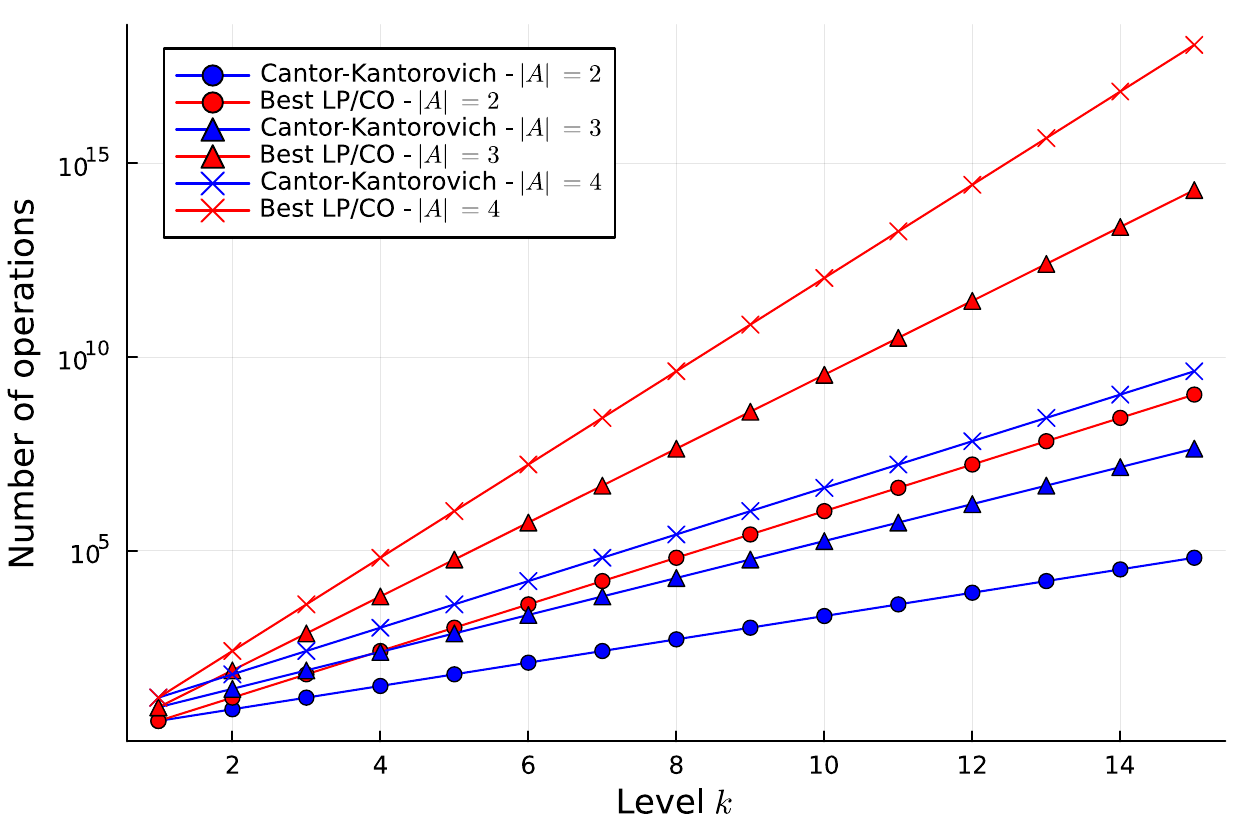}
    \caption{Comparaison of the computational complexitiy of two methods to solve \eqref{eq:kant_LP} for $|A| = 2, 3, 4$. In red, the best LP/CO method, and in blue our method as described in Algorithm~\ref{alg:kantorovich}.}
    \label{fig:complexity}
\end{figure}
\fi 

\section{Application: CK metric for abstractions} \label{sec:application}
In this section, we apply the \textsc{Refine} procedure described in Algorithm~\ref{alg:adaptive} with the Cantor-Kantorovich metric $\mathsf{CK}$ to abstract a given dynamical system. 

\subsection{Electron subject to Lorentz force} \label{sec:lorentz_description} 

We are interested in the position of an electron subject to the Lorentz force (see e.g. \cite{grant2013electromagnetism}). The law is given by 
\begin{equation} \label{eq:lorentz1}
    m\begin{pmatrix}
        \dot{v}_1(t)\\\dot{v}_2(t)\\\dot{v}_3(t)
    \end{pmatrix}
    =
    q \begin{pmatrix} 
        E_1 + B_3 v_2(t) - B_2 v_3(t) \\ 
        E_2 + B_1 v_3(t) - B_3 v_1(t) \\ 
        E_3 + B_2 v_1(t) - B_1 v_2(t)
    \end{pmatrix},
\end{equation}
where $v_1(t)$, $v_2(t)$ and $v_3(t)$ are respectively the $x$-axis, $y$-axis and $z$-axis components of the velocity of the electron in [m/s]. Every other constant and their unities are given in Table~\ref{tab:variables}.
\begin{table}[ht!]
    \centering
    \footnotesize
    \begin{tabular}{llll}
    Symbol & Physical quantity & Unit & Value  \\ \hline
    $m$ & Mass & [kg] & $9.1 \times 10^{-31}$ \\
    $q$ & Charge & [C] & $1.6 \times 10^{-19}$ \\
    $E_1$ & $x$-axis electric field & [V/m] & $-1.0 \times 10^{-10}$ \\
    $E_2$ & $y$-axis electric field & [V/m] & $5.0 \times 10^{-11}$ \\
    $E_3$ & $z$-axis electric field & [V/m] & $0$ \\
    $B_1$ & $x$-axis magnetic field & [T] & $0$ \\
    $B_2$ & $y$-axis magnetic field & [T] & $0$ \\
    $B_3$ & $z$-axis magnetic field & [T] & $1.0 \times 10^{-11}$ \\
    \end{tabular}
    \caption{Constants in the Lorentz force equation \eqref{eq:lorentz1}}
    \label{tab:variables}
\end{table}

Since $B_1 = B_2 = E_3 = 0$, then $v_3(t) = 0$ and the dynamical system can be written on a 2D-plane. Moreover, we are interested in the position of the electron, rather than its velocity: hence, the dynamical equations become the description of a $4$-dimensional affine dynamical system given by
\begin{equation} \label{eq:continuous_affine}
\begin{aligned}
    \begin{pmatrix}
        \dot{p}_1(t) \\
        \dot{p}_2(t) \\
        \dot{v}_1(t) \\
        \dot{v}_2(t) 
    \end{pmatrix}
    &= 
    \begin{pmatrix}
        0 & 0 & 1 & 0\\
        0 & 0 & 0 & 1\\
        0 & 0 & 0 & qB_3/m \\
        0 & 0 & -qB_3/m & 0
    \end{pmatrix}
    \begin{pmatrix}
        p_1(t) \\
        p_2(t) \\
        v_1(t) \\
        v_2(t) 
    \end{pmatrix} \\
    & + 
    \begin{pmatrix}
        0 \\ 0 \\
        qE_1 / m \\
        qE_2 / m
    \end{pmatrix}, 
\end{aligned}
\end{equation}
where $p_1(t)$ and $p_3(t)$ respectively denote the $x$-axis and $y$-axis components of the position of the electron.

We study this dynamical system in discrete time. To do that, we approximate the derivative with the explicit Euler scheme $\dot{x}(t) \approx (x_{t+1} - x_t)/h$, with $h = 0.1$, which gives $x_{t+1} = (I + hA)x_t + hb$, with $x_t = (p_{1, t}, p_{2, t}, v_{1, t}, v_{2, t}) \in \mathbb{R}^4$, $I \in \mathbb{R}^{4 \times 4}$ the identity matrix, and $A \in \mathbb{R}^{4 \times 4}$ and $b \in \mathbb{R}^4$ respectively the matrix and vectors of the continous-time affine system given in \eqref{eq:continuous_affine}. We consider an obstacle defined by $O = [0.5, 1.5] \times [-0.5, 0.5] \times \mathbb{R} \times \mathbb{R}$. The output function $h$ is defined as
\begin{equation}
    h(x_t) = \begin{cases}
        0 &\text{if } x_t \in O, \\
        1 &\text{else if } p_{1, t}, \geq 1.5 \\
        2 &\text{otherwise.}
    \end{cases}
\end{equation}
Finally, we take the uniform measure on $[-1, 4] \times [-1, 1] \times [-1, 1] \times [-1, 1]$ for sampling the initial state, which defines $\lambda$. An illustration of the dynamical system $\Sigma$ is given in Figure~\ref{fig:lorentz}.

\if\useieeelayout1
\begin{figure}[t]
    \centering
    \includegraphics[width = 0.8\linewidth]{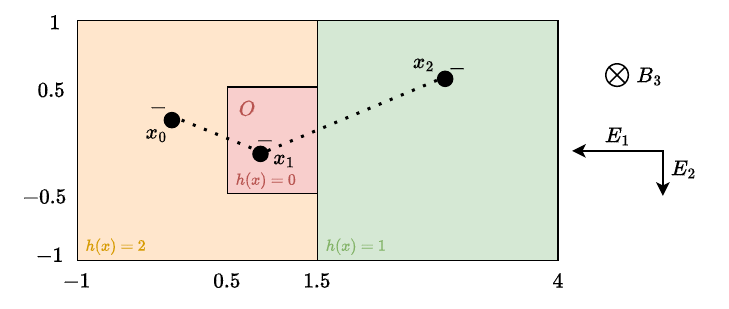}
    \caption{Illustration of the Lorentz force dynamical system $\Sigma$ as defined in Section~\ref{sec:lorentz_description}.}
    \label{fig:lorentz}
\end{figure}
\else
\begin{figure}[ht!]
    \centering
    \includegraphics[width = 0.6\linewidth]{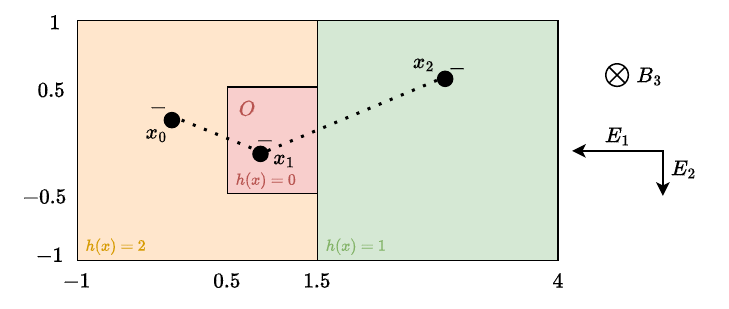}
    \caption{Illustration of the Lorentz force dynamical system $S$ as defined in Section~\ref{sec:lorentz_description}.}
    \label{fig:lorentz}
\end{figure}
\fi 

\subsection{Abstraction-based analysis}
In this section, we are interested in approximating the initial safe set measure 
\begin{equation}
    P_H = \lambda(\{x_0 \in \mathbb{R}^4 : x_t \in \mathbb{R}^4 \setminus O \, \forall t = 0, \dots, H\}), 
\end{equation}
where $H \in \mathbb{N}$ is a given horizon. We show that our procedure $\textsc{Refine}$ with the Cantor-Kantorovich metric $\mathsf{CK}$ yields better results than explicit grid-based approaches. In all of our numerical experiments, we approximated the Cantor-Kantorivich metric with the $\textsc{Cantor-Kantorovich}$ algorithm descripted in Algorithm~\ref{alg:kantorovich} with $\varepsilon = 10^{-3}$.

We proceed as follows. We are given a labeled Markov chain $\Gamma = (S, A, \tau, \mu, l)$ abstracting the dynamical system where each state $s \in S$ corresponds to a block of a partition of the state space. If a block has a non-empty intersection with the obstacle $O$, we label the corresponding state as unsafe, and the other ones as safe. The set of safe states is denoted as $S_{\text{safe}}$. To approximate $P_H$, for a given level of confidence $\beta \in (0, 1)$, we identify the set of states from which the probability of a safe random walk is greater than $1 - \beta$, noted $S_{\beta}$ and formally defined as the set of $s_0 \in S$ such that 
\begin{equation}
    \sum_{s_1 \in S_{\text{safe}}} \tau_{s_0, s_1} \dots \sum_{s_H \in S_{\text{safe}}} \tau_{s_{H-1}, s_H} \geq 1 - \beta.
\end{equation}
Then, for a given $\beta$, $P_H$ is approximated as 
\begin{equation}
    P_H \approx \sum_{s_0 \in S_\beta} \mu_{s_0}.
\end{equation}

We compare two different approaches. In the first one, we uniformly grid the state space. If each dimension is divided into $p$ parts, there will be $p^4$ blocks, noted $B_1, \dots, B_{p^4}$. From this partition, we compute a labeled Markov chain where the states correspond to the $p^4$ blocks with initial measure $\mu_{B_i} = \lambda(B_i)$, and where the transition probability from $B_i$ to $B_j$ is 
\begin{equation}
    P_{B_i, B_j} = \frac{
        \lambda(f^{-1}(B_j) \cap B_i)
    }{\lambda(B_i)}.
\end{equation}
We consider two labeled Markov chains constructed in this way, respectively with $p = 2$ (16 states) and $p = 3$ (81 states). The second approach is to compute the adaptive memory abstraction $\Gamma_W$ with $W = \textsc{Refine}(\Sigma, \lambda, \mathsf{CK}, N)$. Again we study two cases, $N = 6$ (15 states) and $N = 14$ (31 states). An illustration of the two approaches can be found in Figure~\ref{fig:smart_vs_naive}.

\if\useieeelayout1
\begin{figure}[t]
    \centering
    \begin{subfigure}[b]{\linewidth}
        \centering
        \includegraphics[width=0.7\textwidth]{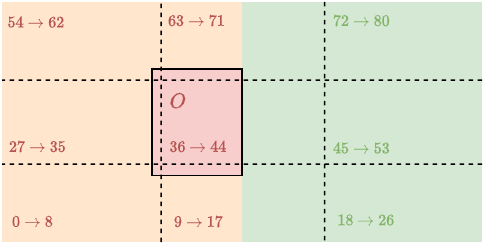}
        \caption{Partition with classical approach.}
        \label{fig:naive}
    \end{subfigure}
    \\
    \begin{subfigure}[b]{\linewidth}
        \centering
        \includegraphics[width=0.7\textwidth]{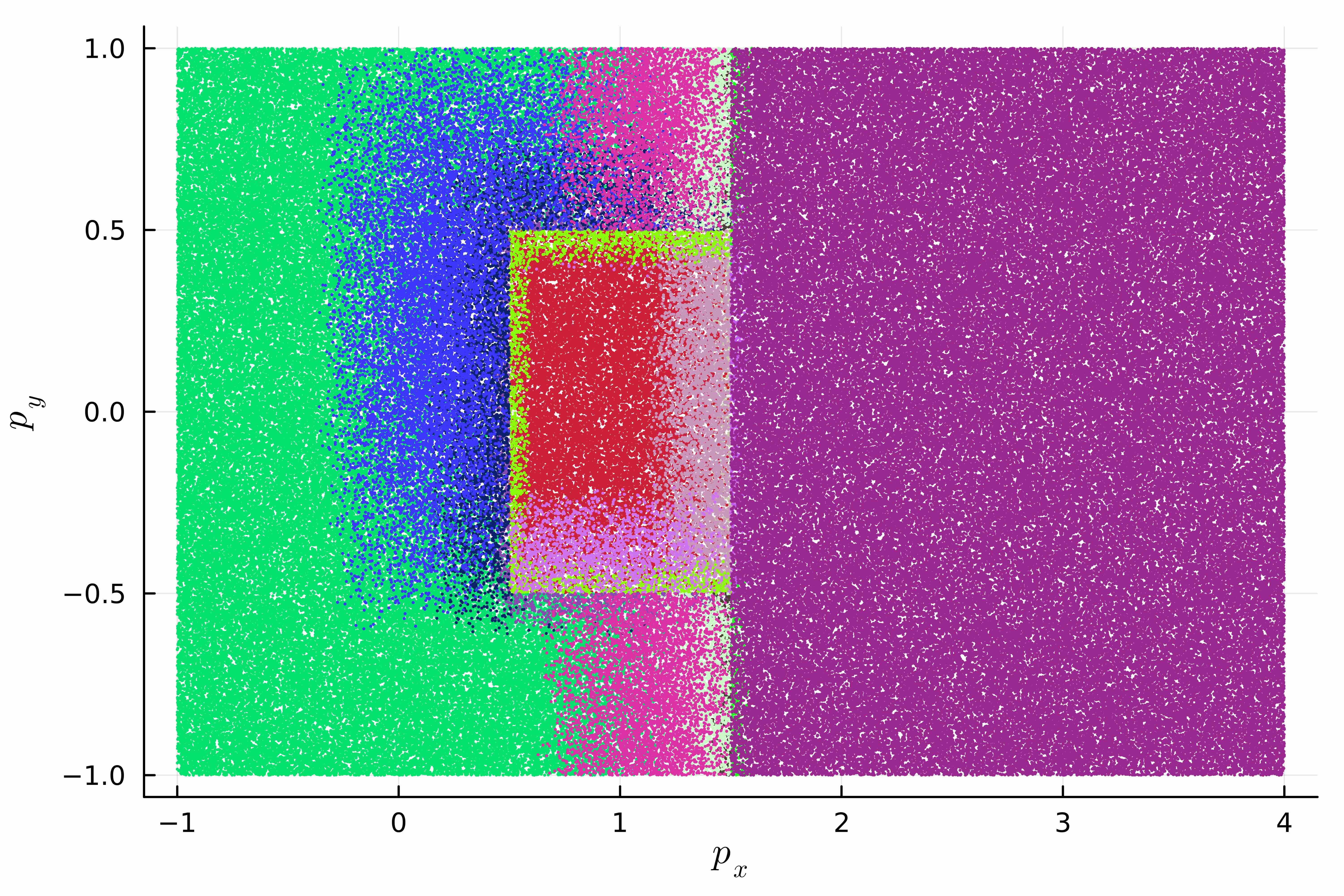}
        \caption{Partition with Algorithm~\ref{alg:adaptive}.}
        \label{fig:smart}
    \end{subfigure}
    \caption{Comparison between a classical partition and the partition computed with Algorithm~\ref{alg:adaptive} on a four dimensional example (here projected on $(p_x, p_y)$). Above: the classical approach, where each dimension is divided into three parts, which gives a total of $3^4 = 81$ blocks. Below: our adaptive memory abstraction with $N=6$, which leads to $15$ blocks in the partition.  Since the algorithm does not explicitely compute a closed form expression for the blocks, we  illustrate them by coloring samples. While there are less blocks here than the $81$ blocks in the above partition, one sees that they are computed smartly, and the critical places in the system (close to the obstacles) benefit from a finer resolution than the large regions away from the obstacle.}
    \label{fig:smart_vs_naive}
\end{figure}
\else
\begin{figure}[ht!]
    \centering
    \begin{subfigure}[b]{0.6\linewidth}
        \centering
        \includegraphics[width=0.7\textwidth]{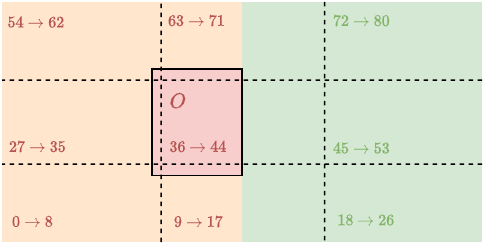}
        \caption{Partition with classical approach.}
        \label{fig:naive}
    \end{subfigure}
    \\
    \begin{subfigure}[b]{0.6\linewidth}
        \centering
        \includegraphics[width=0.8\textwidth]{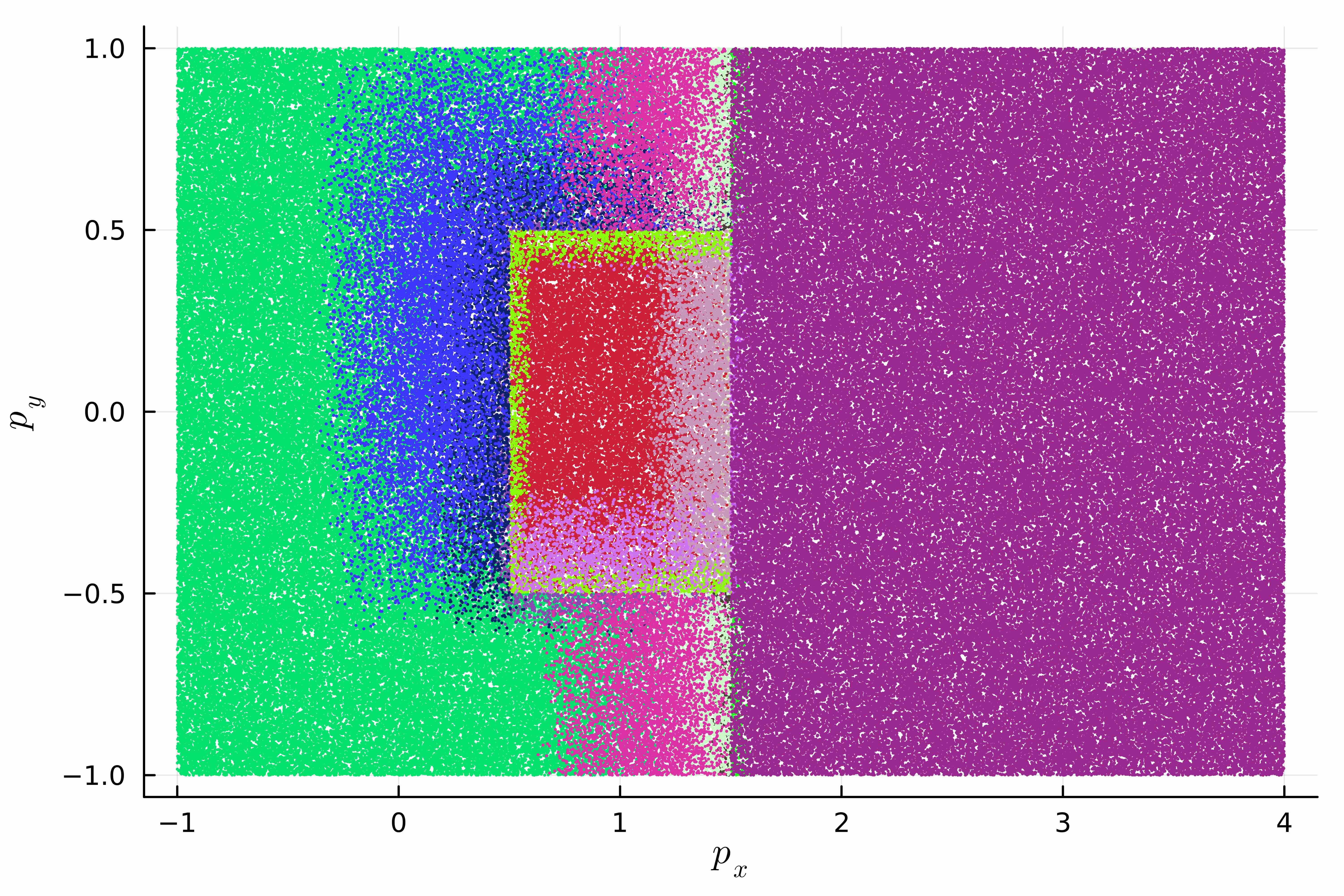}
        \caption{Partition with Algorithm~\ref{alg:adaptive}.}
        \label{fig:smart}
    \end{subfigure}
    \caption{Comparison between a classical partition and the partition computed with Algorithm~\ref{alg:adaptive} on a four dimensional example (here projected on $(p_x, p_y)$). Above: the classical approach, where each dimension is divided into three parts, which gives a total of $3^4 = 81$ blocks. Below: our adaptive memory abstraction with $N=6$, which leads to $15$ blocks in the partition.  Since the algorithm does not explicitely compute a closed form expression for the blocks, we  illustrate them by coloring samples. While there are less blocks here than the $81$ blocks in the above partition, one sees that they are computed smartly, and the critical places in the system (close to the obstacles) benefit from a finer resolution than the large regions away from the obstacle.}
    \label{fig:smart_vs_naive}
\end{figure}
\fi

We compare the approximation of $P_H$ with the four abstractions descripted above, and for different confidence levels. The approximation results and the true probability $P_H$ can be found in Figure~\ref{fig:beta}. We start by discussing the results of the abstractions generated by the \textsc{Refine} algorithm (green and orange lines on Figure~\ref{fig:beta}). When the asked level of confidence is too strict, for example $\beta = 0.01$, then one needs more states to well approximate $P_H$, as we can see on Figure~\ref{fig:beta0.01}. On the other hand, when the level of confidence is too high, for example $\beta = 0.25$, then it is more likely to to overapproximate $P_H$, as one can see on Figure~\ref{fig:beta0.25} for $H = 4$ with the 15 states abstraction. For this example, the confidence level is reasonable when it is equal to $0.05$, as one can see on Figure~\ref{fig:beta0.05}. In this case, the 31 states abstraction perfectly approximates $P_H$ until $H = 8$. Now, one can see that the approximation is better using a the \textsc{Refine} algorithm compared to the na\"ive approach. Indeed the former with 15 states always estimates better $P_H$ than the latter with 81 states. 

\if\useieeelayout1
\begin{figure}[t]
    \centering
    \begin{subfigure}[b]{\linewidth}
        \centering
        \includegraphics[width=0.8\textwidth]{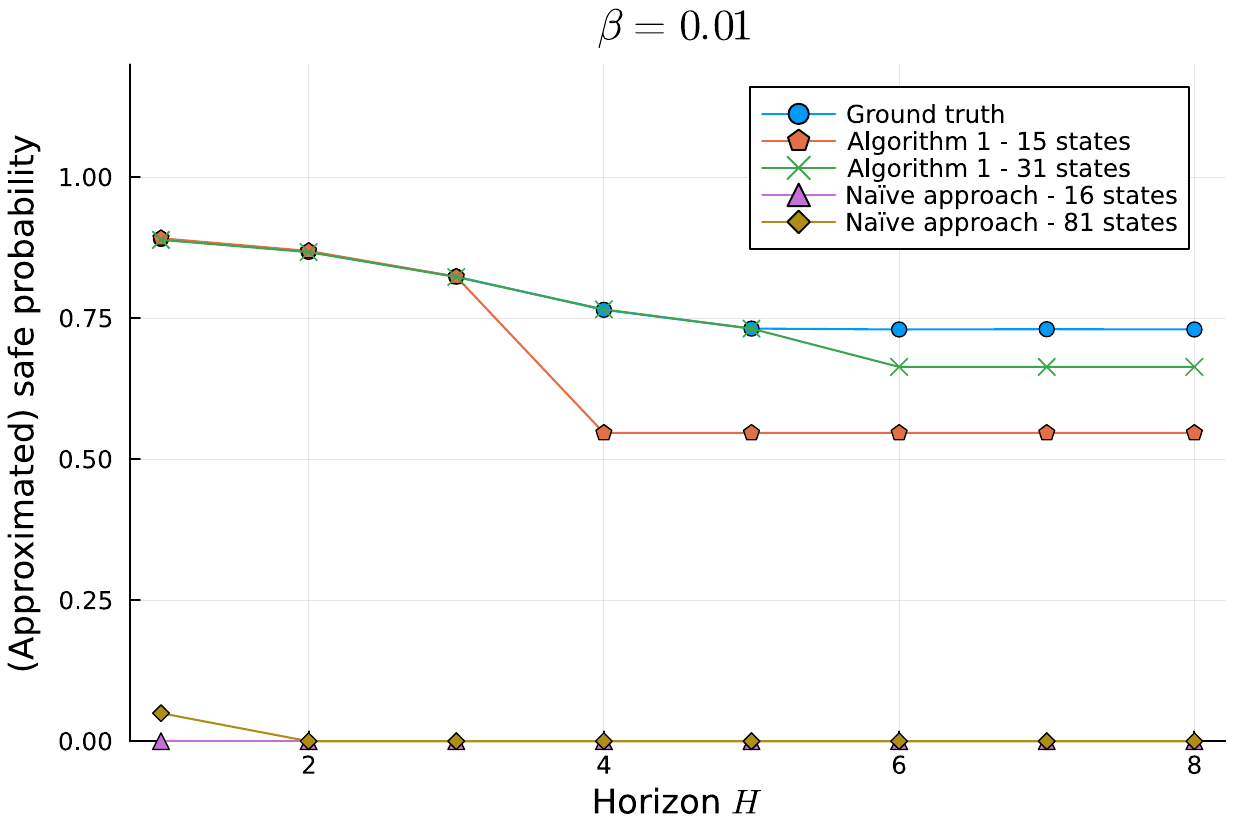}
        \caption{Confidence level $\beta = 0.01$. The confidence level is too low, but with 31 states the \textsc{Refine} abstraction still approximates well $P_H$.}
        \label{fig:beta0.01}
    \end{subfigure}
    \\
    \begin{subfigure}[b]{\linewidth}
        \centering
        \includegraphics[width=0.8\textwidth]{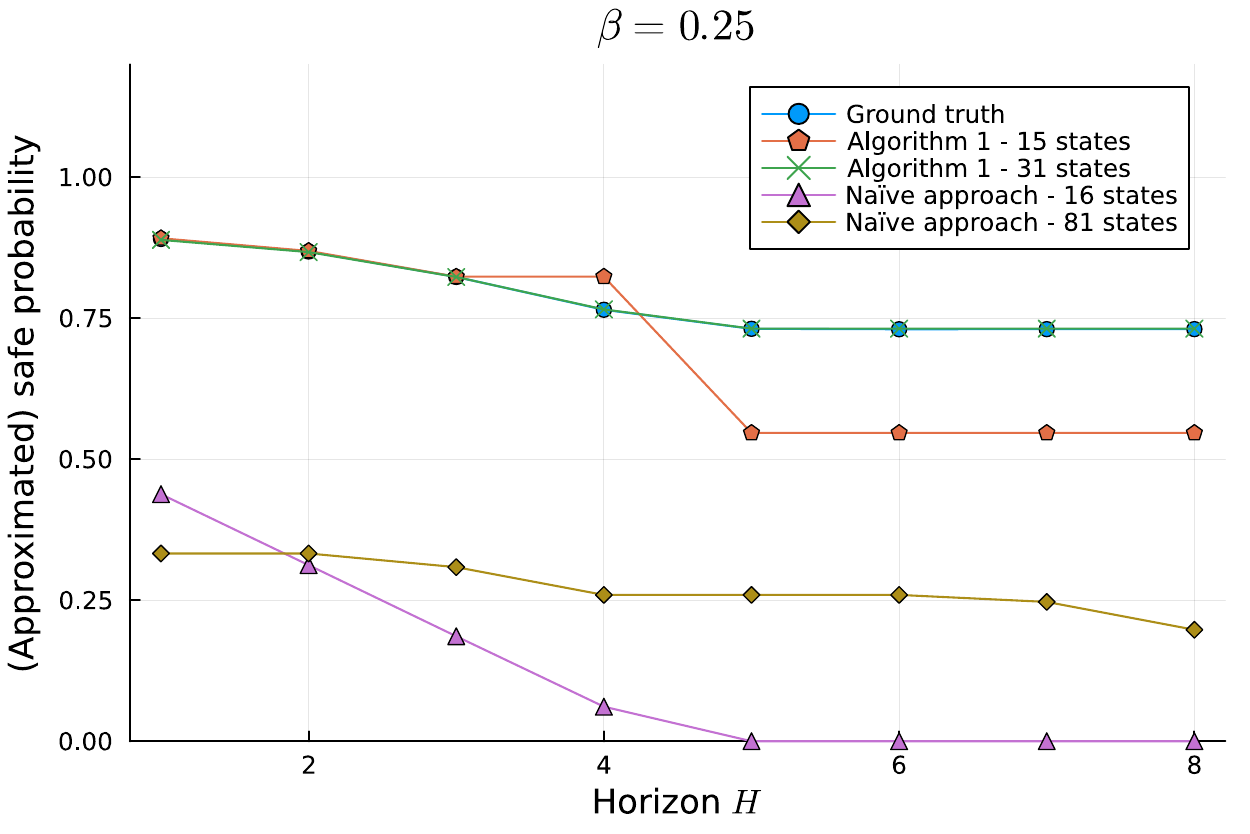}
        \caption{Confidence level $\beta = 0.25$. The confidence level is too high, and $P_H$ may be overapproximated.}
        \label{fig:beta0.25}
    \end{subfigure}
    \\
    \begin{subfigure}[b]{\linewidth}
        \centering
        \includegraphics[width=0.8\textwidth]{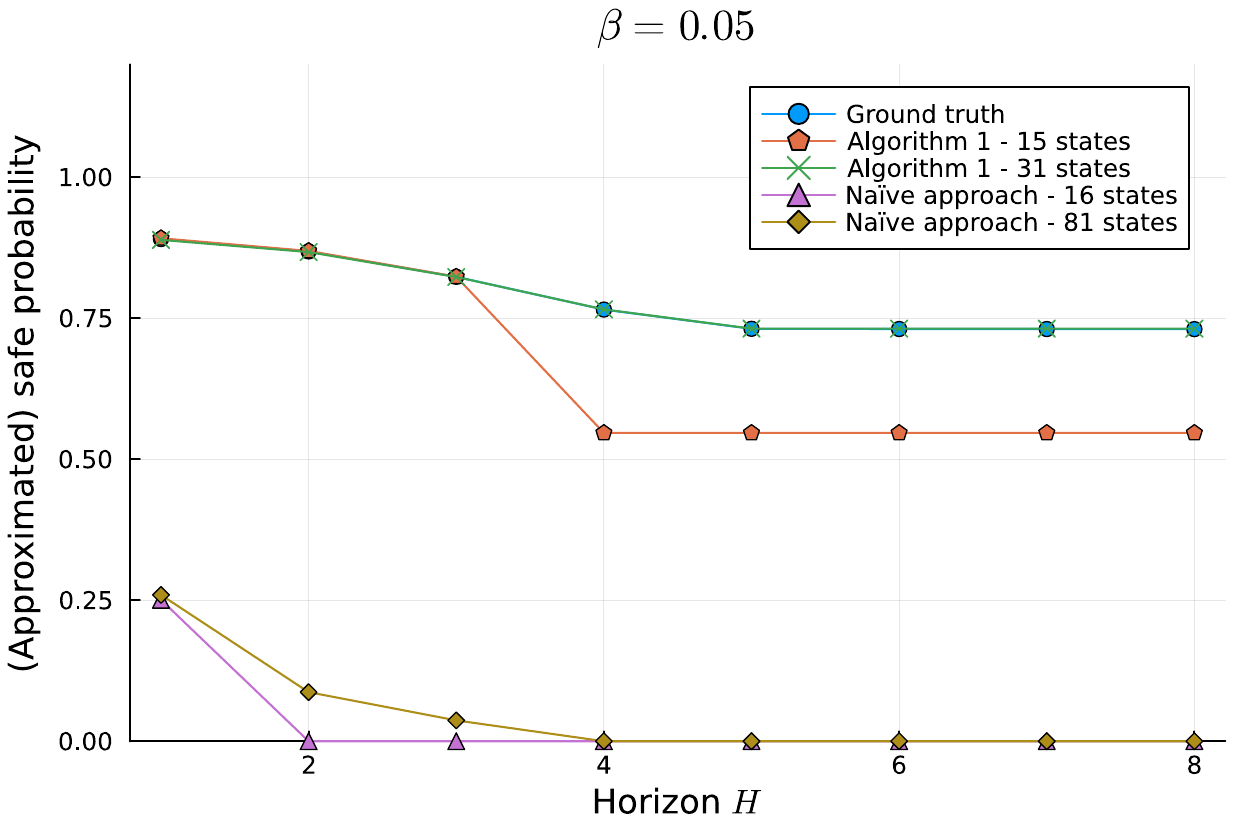}
        \caption{Confidence level $\beta = 0.05$. The confidence level is reasonable, and the \textsc{Refine} abstractions yield good results with both 15 and 31 states.}
        \label{fig:beta0.05}
    \end{subfigure}
    \caption{Approximation of $P_H$ for $H = 0, \dots, 8$ with different abstractions and for different confidence levels $\beta$.}
    \label{fig:beta}
\end{figure}
\else 
\begin{figure}[ht!]
    \centering
    \begin{subfigure}[b]{0.6\linewidth}
        \centering
        \includegraphics[width=0.8\textwidth]{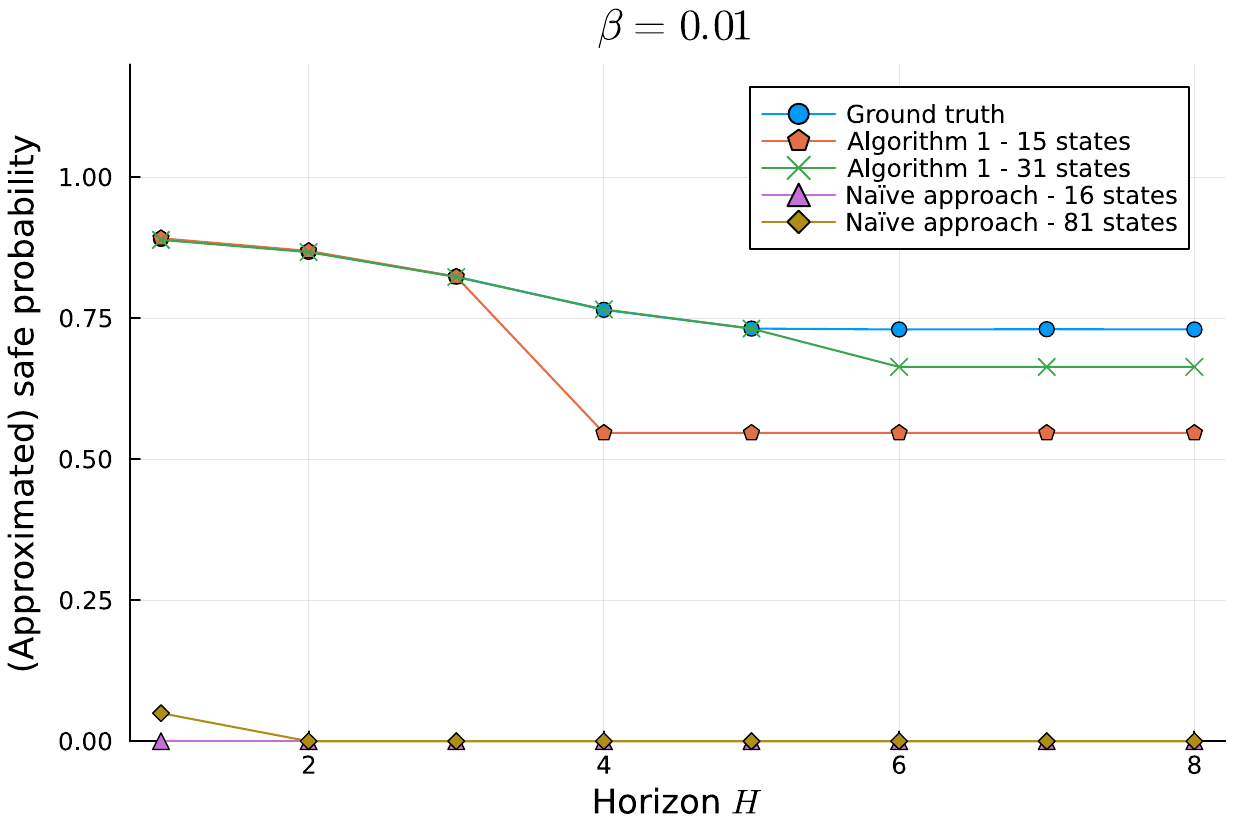}
        \caption{Confidence level $\beta = 0.01$. The confidence level is too low, but with 31 states the \textsc{Refine} abstraction still approximates well $P_H$.}
        \label{fig:beta0.01}
    \end{subfigure}
    \\
    \begin{subfigure}[b]{0.6\linewidth}
        \centering
        \includegraphics[width=0.8\textwidth]{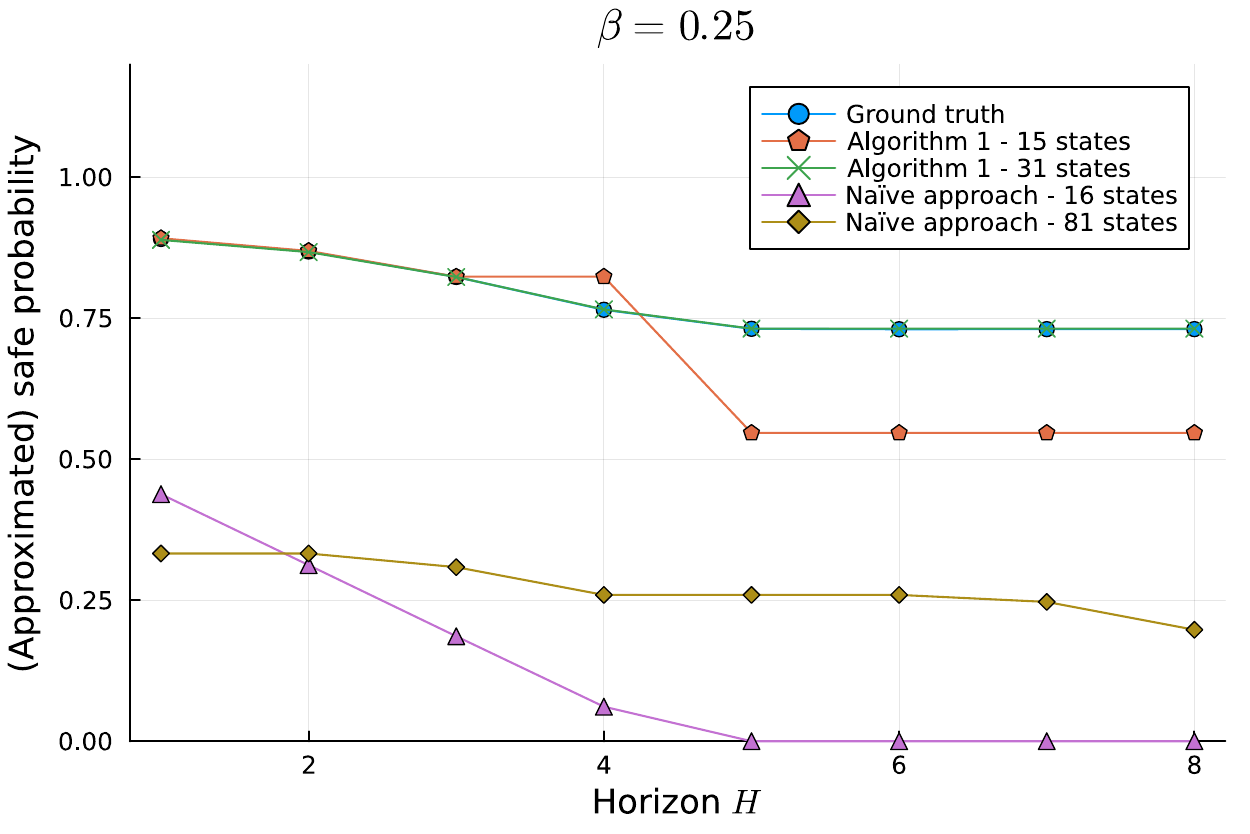}
        \caption{Confidence level $\beta = 0.25$. The confidence level is too high, and $P_H$ may be overapproximated.}
        \label{fig:beta0.25}
    \end{subfigure}
    \\
    \begin{subfigure}[b]{0.6\linewidth}
        \centering
        \includegraphics[width=0.8\textwidth]{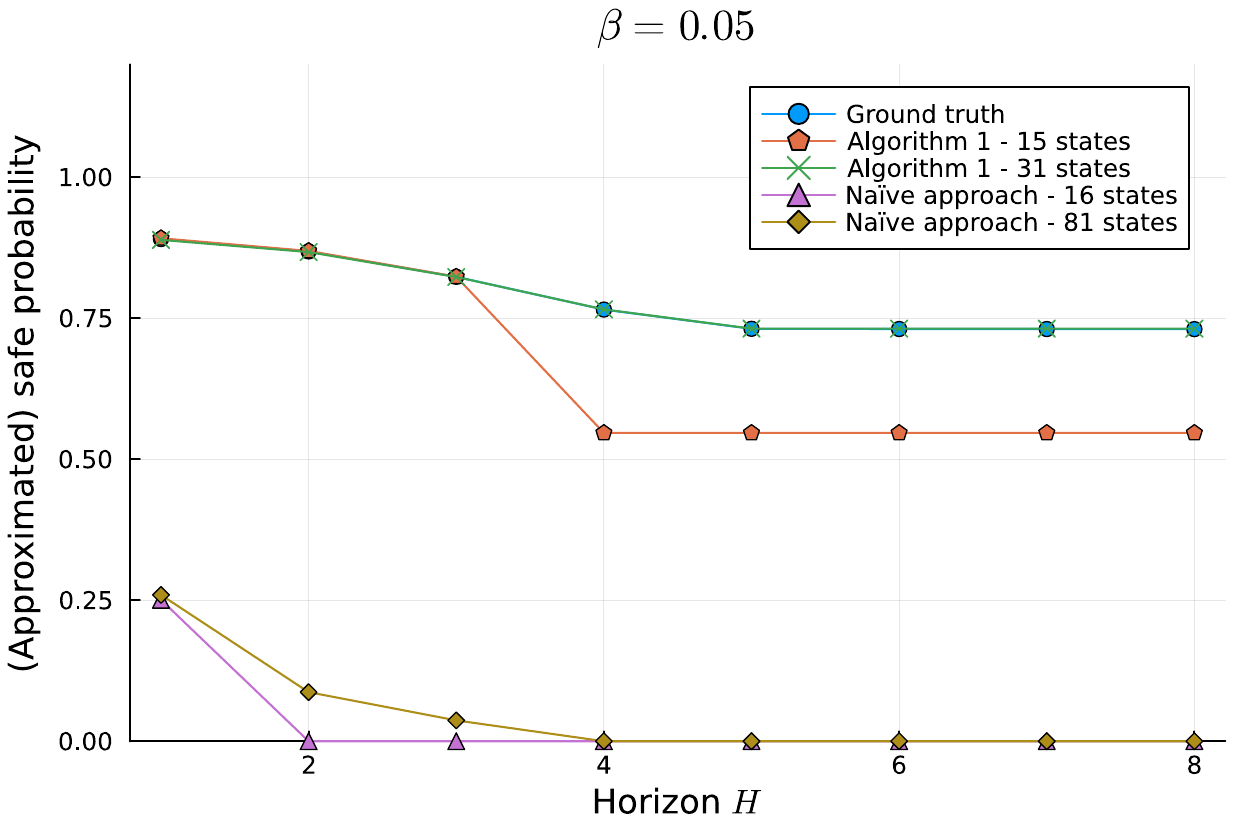}
        \caption{Confidence level $\beta = 0.05$. The confidence level is reasonable, and the \textsc{Refine} abstractions yield good results with both 15 and 31 states.}
        \label{fig:beta0.05}
    \end{subfigure}
    \caption{Approximation of $P_H$ for $H = 0, \dots, 8$ with different abstractions and for different confidence levels $\beta$.}
    \label{fig:beta}
\end{figure}
\fi 

\section{Conclusions} \label{sec:conclusions}

In this work, we have introduced Markov models encoding adaptive-memory schemes to abstract dynamical systems via samples. We have proved that the abstractions  preserve safety properties of the abstracted systems. Along with this, given a notion of metric between Markov models, we have proposed a generic and tuneable procedure to choose a convenient memory scheme. We have then introduced a Cantor-Kantorovich metric for Markov chains, proving  that the latter satisfies positivity, symmetry and triangle inequality, and we have provided an algorithm to approximate it efficiently. We have finally performed an abstraction-based safety analysis on a example corresponding to a real-life dynamical system. We showed that our method yields better approximations than classical, grid-based abstraction  approaches in the litterature. 

For further work, we would like to test our method with other metrics between Markov models, and compare numerical results. For example, one could investigate metrics that encode control specifications, unlike our metric. On the other hand, we would like to improve the algorithm for computing the Cantor-Kantorovich metric. 

\if\useieeelayout1
\appendices
\else
\appendix
\fi

\section{Consequence of Assumption~\ref{ass:mc_well_defined}} \label{app:stochastic_matrix}
\begin{proposition} \label{prop:mc_well_defined}
    Given a dynamical system $\Sigma$, a measure $\lambda$ and a set $W$ that defines a partition, if Assumption~\ref{ass:mc_well_defined} is satisfied, then the adaptive memory abstraction $\Gamma_W = (W, A, \tau, \mu, l)$ has a stochastic transition matrix $\tau$.
\end{proposition}
\begin{proof}
    By \eqref{eq:cond_MC} and Assumption~\ref{ass:mc_well_defined}, it suffices to show that, for all $w_{[-s, t]} \in W$, 
    \begin{equation} \label{eq:ltp_toprove}
        \sum_{w'_{[-s', t']} \in W}
        \lambda\left(f^{-1}\left(\left[ w'_{[-s', t']} \right]\right) \cap \left[ w_{[-s, t]} \right]\right)
        = \lambda\left(\left[ w_{[-s, t]} \right]\right)
    \end{equation}
    is satisfied. Since the sets $f^{-1}\left(\left[w_{[-s, t]}\right]\right)$ also form a partition of $X$, \eqref{eq:ltp_toprove} holds by the law of total probability, and the proof is completed.
\end{proof}

\section{Procedure $\textsc{Refine}$ yields a partition} \label{app:procedure}
\begin{assumption} \label{ass:CDC_ass}
    The dynamical system $\Sigma$ and measure $\lambda$ are such that, for all $k$, for all $w \in A^k$, if $\lambda([w]_{[0, k-1]}) = 0$, then $[w]_{[0, k-1]} = \emptyset$.
\end{assumption}
\begin{proposition}
    Suppose that Assumption~\ref{ass:CDC_ass} holds. Let $\Sigma$ be a dynamical system with initial measure $\lambda$, and let $\mathsf{D}$ be a metric between labeled Markov chains. Then, for all $N \in \mathbb{N}_{> 0}$ the output of $\textnormal{\textsc{Refine}}(\Sigma, \lambda, \mathsf{D}, N)$ defines a partition of $X$ and the matrix of the corresponding abstraction is stochastic.
\end{proposition}
\begin{proof}
    By Proposition~\ref{prop:mc_well_defined}, we have to show that each current model $\Gamma_W$ is such that $W$ defines a partition of $X$ and satisfies Assumption~\ref{ass:mc_well_defined}. We prove the claim by induction. 

    First, we prove that the set $W$ such as defined after line~\ref{line:4} defines a partition and satisfies Assumption~\ref{ass:mc_well_defined}. The set $W = \{(a)_{[0, 0]}\}_{a \in A}$ defines a partition by definition of the output function $h$. Now, the set $W \setminus \{(a)_{[0, 0]}\}_{a \in A_0}$ satisfies Assumption~\ref{ass:mc_well_defined} by construction. Moreover, it still defines a partition because the sets $[a_{[0, 0]}]$ for $a \in A_0$ are all empty by Assumption~\ref{ass:CDC_ass}.

    We now assume that $W$ defines a partition and satisfies Assumption~\ref{ass:mc_well_defined}. We prove that, for any $w^i_{[0, t_i]} \in W$, the set $W'_i$ such as after \ref{line:12} is a partition and satisfies Assumption~\ref{ass:mc_well_defined}. We first show that the expanded set $\tilde{W} = (W \setminus \{w^i_{[0, t_i]}\} ) \cup \{(w^ia)_{[0, t_i + 1]}\}_{a \in A}$ defines a partition. Condition~\ref{eq:partition_cond_1} is satisfied because
    \begin{equation}
    \begin{aligned}
        &\bigcup_{w_{[0, t]} \in \tilde{W}} \left[w_{[0, t]}\right] \\
        =& \,  
        \left(\bigcup_{w_{[0, t]} \in W \setminus \left\{w^i_{[0, t_i]}\right\}} \left[w_{[0, t]}\right]\right) \cup \left(\bigcup_{a \in A} \left[(w^ia)_{[0, t_i+1]}\right]\right)\\
        =& \, 
        \left( \bigcup_{w_{[0, t]} \in W \setminus \left\{w^i_{[0, t_i]}\right\}} \left[w_{[0, t]}\right] \right) \cup \left[w^i_{[0, t_i]}\right], 
    \end{aligned}
    \end{equation}
    which is the whole state space $X$ by the recursion assumption. And condition~\ref{eq:partition_cond_2} is also satisfied because $\left[(w^ia)_{[0, t_i + 1]}\right] \subseteq \left[w^i_{[0, t_i]}\right]$ for all $a \in A$. Now, with $W_0$ as defined in line~\ref{line:11}, the $\tilde{W} \setminus W_0$ satisfies Assumption~\ref{ass:mc_well_defined} by construction. Finally, it still defines a partition because the sets $[w_{[0, t]}]$ for $w_{[0, t]} \in W_0$ are empty by Assumption~\ref{ass:CDC_ass}, and the proof is completed.
\end{proof}

\section{Proof of Lemma~\ref{lemma:firstlemma}} \label{app:proof_firstlemma}
We first prove that $\pi^k(w, w) \leq \min\{p_1^k(w), p_2^k(w)\}$. Constraints \eqref{eq:cons_LP_marg} imply that 
\begin{equation}
\begin{aligned}
p_1^k(w) = \pi^k(w, w) + \sum_{\substack{w' \in A^k\\ w \neq w'}} \pi^k(w, w') \geq \pi^k(w, w), \\
p_2^k(w) = \pi^k(w, w) + \sum_{\substack{w' \in A^k\\ w \neq w'}} \pi^k(w', w) \geq \pi^k(w, w), 
\end{aligned}
\end{equation}
which implies $\pi^k(w, w) \leq \min\{p_1^k(w), p_2^k(w)\}$. We now prove that $\pi^k(w, w) \geq \min\{p_1^k(w), p_2^k(w)\}$. By contradiction, let $\pi^k$ be an optimal solution to \eqref{eq:kant_LP} such that there exists $w \in A^k$ and $\varepsilon > 0$ with
$
\pi^k(w, w) = \min\{p_1^k(w), p_2^k(w)\} - \varepsilon.
$
Assume without loss of generality that $\min\{p_1^k(w), p_2^k(w)\} = p_1^k(w)$, then constraints \eqref{eq:cons_LP_marg} imply that 
\begin{enumerate}
\item there exists $w' \neq w$ such that $\pi^k(w, w') = \varepsilon'$ for some $\varepsilon' \in (0, \varepsilon]$, and 
\item there exists $w'' \neq w$ such that $\pi^k(w'', w) = \varepsilon''$ for some $\varepsilon'' \in (0, \varepsilon]$. 
\end{enumerate}
Let $\mathsf{K}_{\mathsf{C}}(p_1^k, p_2^k)$ denote the Kantorovich metric corresponding to such $\pi^k$. Now assume, again without loss of generality, that $\varepsilon' \leq \varepsilon''$. Consider then $(\pi^k)'$ such that $(\pi^k)'(w_1, w_2) = \pi^k(w_1, w_2)$ for all $w_1, w_2 \in A^k$ except 
\begin{enumerate}
\item $(\pi^k)'(w, w) = \pi^k(w, w) + \varepsilon'$, 
\item $(\pi^k)'(w, w') = \pi^k(w, w') - \varepsilon'$, 
\item $(\pi^k)'(w'', w') = \pi^k(w'', w') + \varepsilon'$, 
\item  $(\pi^k)'(w'', w) = \pi^k(w'', w) - \varepsilon'$.
\end{enumerate}
The joint distribution $(\pi^k)'$ is feasible since it satisfies constraints \eqref{eq:cons_LP_pos} and \eqref{eq:cons_LP_marg}. Now let $\mathsf{K}_{\mathsf{C}}'(p_1^k, p_2^k)$ denote the solution corresponding to such $(\pi^k)'$, it holds that 
\begin{equation}
\begin{aligned}
&\mathsf{K}_{\mathsf{C}}'(p_1^k, p_2^k) = \mathsf{K}_{\mathsf{C}}(p_1^k, p_2^k) \\
&\quad - \varepsilon'[\mathsf{C}(w, w') + \mathsf{C}(w', w'') - \mathsf{C}(w', w'')].
\end{aligned}
\end{equation}
By triangular inequality of $\mathsf{C}$, we have that $\mathsf{K}_{\mathsf{C}}'(p_1^k, p_2^k) < \mathsf{K}_{\mathsf{C}}(p_1^k, p_2^k)$, which contradicts the fact that $\pi^k$ is optimal, and the proof is completed. \hfill $\blacksquare$

\section{Proof of Lemma~\ref{lemma:secondlemma}} \label{app:proof_secondlemma}
For some $w \in A^{k-1}$, assume without loss of generality that $p_1^{k-1}(w) > p_2^{k-1}(w)$. First, by the constraints \eqref{eq:cons_LP_marg}, it holds that
\begin{equation}
    \sum_{\substack{w' \in A^{k-1}\\w' \neq w}} \sum_{a_1, a_2 \in A} \pi^k(wa_1, w'a_2) \geq p_1^{k-1}(w) - p_2^{k-1}(w). 
\end{equation}
Now, we proceed similarly as for Lemma~\ref{lemma:firstlemma}. Suppose by contradiction that 
$
    \sum_{\substack{w' \in A^{k-1}\\w' \neq w}} \sum_{a_1, a_2 \in A} \pi^{k}(wa_1, w'a_2) > p_1^{k-1}(w) - p_2^{k-1}(w). 
$
Then there exists $w' \neq w \in A^{k-1}$, and $a_1, a_2 \in A$ such that $\pi^{k}(w'a_1, wa_2) = \varepsilon' > 0$. There also exists $w'' \in A^{k-1}$ such that $w'' \neq w$ and $w'' \neq w'$, and $a_3, a_4 \in A$ such that $\pi^{k}(wa_3, w''a_4) = \varepsilon'' > 0$. Assume w.l.o.g. that $\varepsilon' \leq \varepsilon''$, and consider a solution $(\pi^{k})'$ such that $(\pi^{k})' = \pi^{k}$, except for 
\begin{enumerate}
\item $(\pi^{k})'(wa_3, w''a_4) = \pi^{k}(wa_3, w''a_4) - \varepsilon'$, 
\item $(\pi^{k})'(w'a_1, wa_2) = \pi^{k}(w'a_1, wa_2) - \varepsilon'$, 
\item $(\pi^{k})'(w'a_1, w''a_4) = \pi^{k}(w'a_1, w''a_4) + \varepsilon'$, and
\item $(\pi^{k})'(wa_3, wa_2) = \pi^{k}(wa_3, wa_2) + \varepsilon'$.
\end{enumerate} 
The joint distribution $(\pi^{k})'$ is feasible since it still satisfies the constraints~\eqref{eq:cons_LP_pos} and \eqref{eq:cons_LP_marg}. Note that, since the Cantor distance satisfies the strong triangular inequality (see Lemma~\ref{lemma:ultra}), 
\begin{equation}
\begin{aligned}
    \mathsf{C}(w'a_1, w''a_4) &\leq \max\{\mathsf{C}(w'a_1, wa_2), \mathsf{C}(wa_2, w''a_4)\}\\
    &= \max\{\mathsf{C}(w'a_1, wa_2), \mathsf{C}(wa_3, w''a_4)\}.
\end{aligned}
\end{equation}
Moreover, $\mathsf{C}(wa_3, wa_2) = 2^{-k}$. Now let $\mathsf{K}_{\mathsf{C}}'(p_1^{k}, p_2^{k})$ denote the solution corresponding to such $(\pi^{k})'$, we have that $\mathsf{K}_{\mathsf{C}}(p_1^{k}, p_2^{k}) - \mathsf{K}_{\mathsf{C}}'(p_1^{k}, p_2^{k})$ is
\begin{equation}
\begin{aligned}
    &\mathsf{K}_{\mathsf{C}}(p_1^{k}, p_2^{k}) - \mathsf{K}_{\mathsf{C}}'(p_1^{k}, p_2^{k}) = 
    - \varepsilon'\left[
        \begin{array}{ll}
            +& \mathsf{C}(w'a_1, wa_2) \\
            +& \mathsf{C}(wa_3, w''a_4) \\
            -& \mathsf{C}(w'a_1, w''a_4)\\ 
            -& 2^{-k}
        \end{array}
    \right]\\
    &\quad \leq - \varepsilon'\left[
        \begin{array}{ll}
            +& \mathsf{C}(w'a_1, wa_2) + \mathsf{C}(wa_3, w''a_4) \\
            -&\max\{\mathsf{C}(w'a_1, wa_2), \mathsf{C}(wa_3, w''a_4)\}\\ 
            -& 2^{-k}
        \end{array}
    \right] \\
    &\quad \leq - \varepsilon'[2^{-(k-1)} - 2^{-k}] \\
    &\quad \leq 0, 
\end{aligned}
\end{equation}
which contradicts the fact that $\pi^{k}$ is optimal. \hfill $\blacksquare$

\if\useieeelayout0 \if\review1 \nolinenumbers \fi \fi

\if\useieeelayout1 \section*{References} \fi
\bibliographystyle{ieeetr}
\bibliography{itac23.bib}

\if\useieeelayout1
\begin{IEEEbiography}[{\includegraphics[width=1in,keepaspectratio]{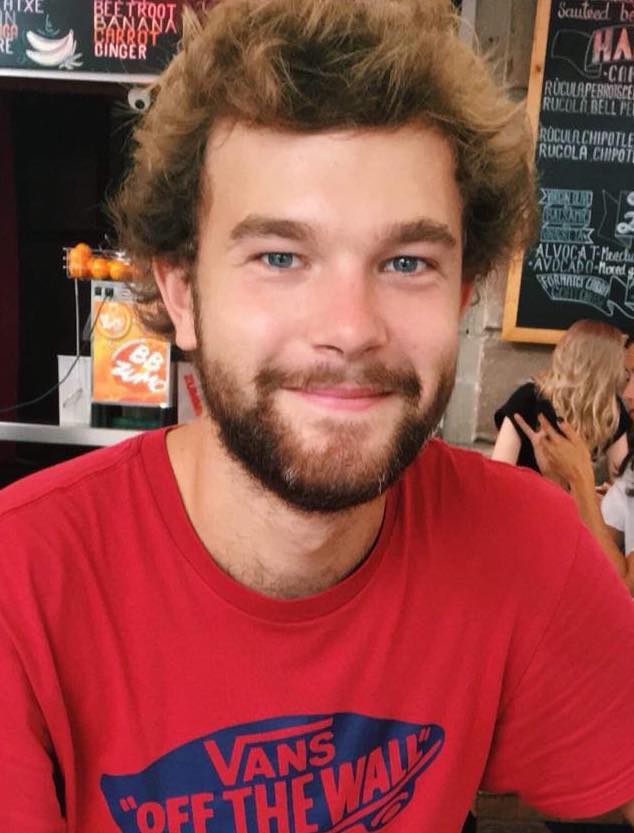}}]{Adrien Banse} received his M.Sc. degree in Mathematical Engineering from UCLouvain, Belgium, in 2022. How is now a Ph.D. candidate at the Department of Mathematical Engineering in the Institute of Information and Communication Technologies (INMA/ICTEAM) at UCLouvain. His research interests lie are at the intersection of Computer Science and Applied Mathematics, including data-driven formal methods and analysis of dynamical systems. He is a FRIA/FNRS fellow. 
\end{IEEEbiography}
\begin{IEEEbiography}[{\includegraphics[width=1in, keepaspectratio]{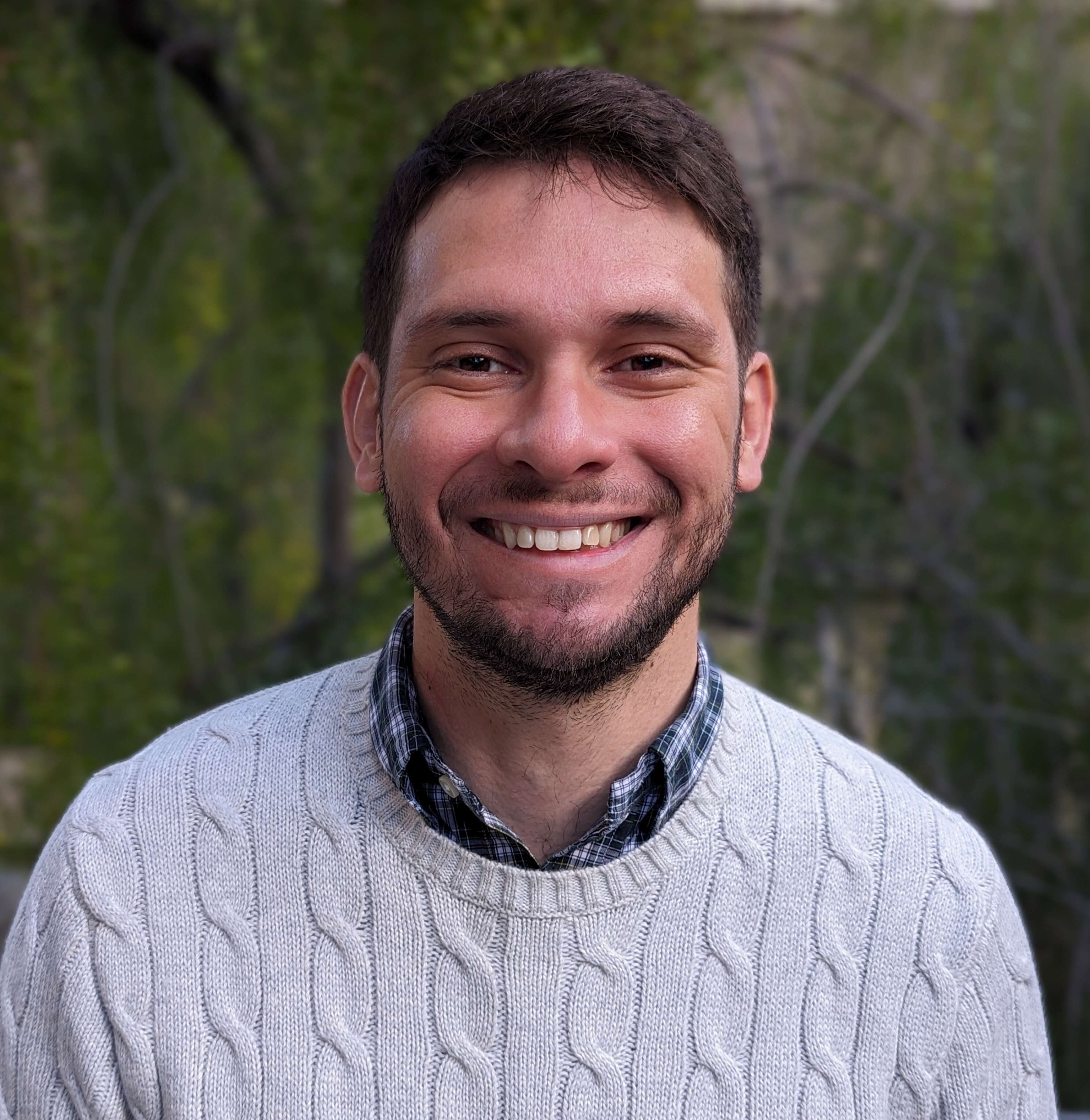}}]{Licio Romao} received the B.Sc. degree in electrical engineering from the Federal University of Campina Grande (UFCG), Campina Grande,
Brazil, in 2014, the M.Sc. degree in electrical engineering from the University of Campinas (UNICAMP), Campinas, Brazil, in 2017, and the D.Phil. (Ph.D.) degree in engineering science from the University of Oxford, Oxford, U.K., in
2021. He is currently a Postdoctoral Researcher with the Department of Aeronautics and Astronautics, Stanford University, USA, and has previously worked as a postdoctoral researcher with the Department of Computer Science, University of Oxford. His research focuses on the design of feedback controllers for stochastic dynamical systems. It leverages tools from formal methods, optimisation, statistics, stochastic calculus, and control theory. He is a recipient of the 2021 IET Control and Automation Doctoral Dissertation prize.

\end{IEEEbiography}
\begin{IEEEbiography}[{\includegraphics[width=1in, keepaspectratio]{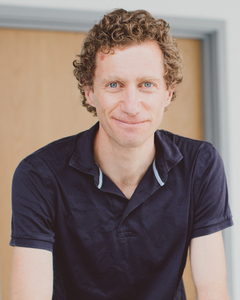}}]{Alessandro Abate} is Professor of Verification and Control in the Department of Computer Science at the University of Oxford. Earlier, he did research at Stanford University and at SRI International, and was an Assistant Professor at the Delft Center for Systems and Control, TU Delft. He received MS and PhD from the University of Padua and UC Berkeley.  
His research interests lie on the formal verification and control of stochastic hybrid systems, and in their applications to cyber-physical systems, particularly involving safety criticality and energy. He blends in techniques from probability theory and logic, machine learning and AI, towards assured autonomy.

\end{IEEEbiography}
\begin{IEEEbiography}[{\includegraphics[width=1in, keepaspectratio]{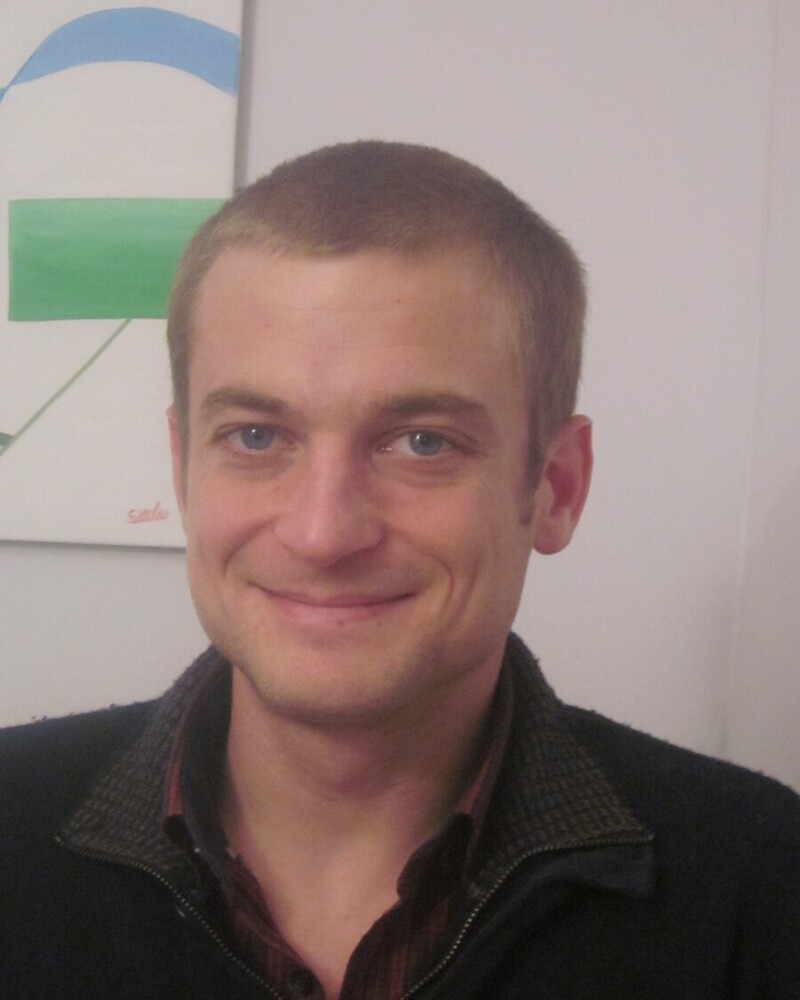}}]{Rapha\"el M. Jungers} is a Professor at UCLouvain, Belgium. His main interests lie in the fields of Computer Science, Graph Theory, Optimization and Control. He received a Ph.D. in Mathematical Engineering from UCLouvain (2008), and a M.Sc. in Applied Mathematics, both from the Ecole Centrale Paris, (2004), and from UCLouvain (2005). He has held various invited positions, at the Universit\'e Libre de Bruxelles (2008-2009), at the Laboratory for Information and Decision Systems of the Massachusetts Institute of Technology (2009-2010), at the University of L’Aquila (2011, 2013, 2016), and at the University of California Los Angeles (2016-2017). He is a FNRS, BAEF, and Fulbright fellow. He has been an Associate Editor for the IEEE CSS Conference Editorial Board, and the journals NAHS, Systems and Control Letters, and IEEE Transactions on Automatic Control. He was the recipient of the IBM Belgium 2009 award and a finalist of the ERCIM Cor Baayen award 2011. He was the co-recipient of the SICON best paper award 2013-2014.
\end{IEEEbiography}
\fi

\end{document}